\documentclass[journal,twoside,web]{ieeecolor}
\pdfoutput=1
\usepackage{generic}
\usepackage{cite}
\usepackage{amsmath,amssymb,amsfonts}
\usepackage{algorithm}
\usepackage[justification=centering]{caption}
\usepackage[noend]{algpseudocode}
\usepackage{algorithmicx}
\usepackage{graphicx}
\usepackage{textcomp}
\usepackage{bm}
\usepackage{color}
\usepackage{mathrsfs}

\newtheorem{remark}{Remark}[section]
\newtheorem{theorem}{Theorem}[section]
\newtheorem{assumption}{Assumption}[section]
\newtheorem{lemma}{Lemma}[section]

\newtheorem{corollary}{Corollary}[section]
\newtheorem{example}{Example}[section]

\allowdisplaybreaks

\def\ban{\begin{eqnarray*}}
	\def\ean{\end{eqnarray*}}
\def\bna{\begin{eqnarray}}
	\def\ena{\end{eqnarray}}
\def\BibTeX{{\rm B\kern-.05em{\sc i\kern-.025em b}\kern-.08em
    T\kern-.1667em\lower.7ex\hbox{E}\kern-.125emX}}
\begin{document}
\title{Distributed adaptive estimation for stochastic large regression models}
\author{Die Gan,  Siyu Xie, Zhixin Liu, \IEEEmembership{Member, IEEE}, and Xuebo Zhang \IEEEmembership{Senior Member, IEEE}
\thanks{Corresponding author: Zhixin Liu.}
\thanks{Die Gan and Xuebo Zhang are with the College of Artificial Intelligence, Nankai University, Tianjin 300350, China, (e-mails: gandie@nankai.edu.cn; zhangxuebo@nankai.edu.cn).
}
\thanks{Siyu Xie is with the School of Aeronautics and Astronautics, University of Electronic Science and Technology of China, Chengdu 611731, China (e-mail: syxie@uestc.edu.cn).}
\thanks{Zhixin Liu is with State Key Laboratory of Mathematical Sciences, Academy of Mathematics and Systems Science, Chinese Academy of Sciences, Beijing 100190, China, (e-mail: lzx@amss.ac.cn)}
}
\maketitle

\begin{abstract}
	This paper studies the distributed adaptive estimation  problems for stochastic large regression models with an infinite number of  parameters. 
	By constructing a recursive local cost function,
	we propose a novel distributed recursive least squares algorithm  to estimate the unknown system parameters, where the growth rate of regressors' dimension is characterized by a non-decreasing positive function. The 
	almost sure convergence   of the proposed algorithm is established
	under  a cooperative  excitation condition, which incorporates the temporal information and the spatial information to reflect the cooperative effect among multiple agents.
	Moreover, we analyze the
	prediction error by establishing  the asymptotic upper bound of the  
	accumulated regret without any excitation conditions.  
	The main difficulty of theoretical analysis lies in how to analyze properties of the product of non-independent and non-stationary random matrices, whose dimensions  change over time simultaneously.
	Some techniques, such as stochastic Lyapunov function, double-array martingale theory and algebraic graph theory,  are employed to deal with the above issue.
	Our theoretical results are derived without imposing independence or stationarity assumptions on the regression vectors, thereby not excluding the correlated feedback signals.

\end{abstract}


\section{Introduction}\label{sec:introduction}
\IEEEPARstart{W}{ith} the rapid growth of computational resources and data availability, large models   have attracted significant attention  due to their remarkable ability to approximate intricate behaviors across diverse domains\cite{wei2022emergent}. 
These models are typically characterized by a vast or even infinite number of system parameters. Modern deep neural networks, such as large language models, often contain millions to trillions of trainable parameters.
Accurately estimating  these parameters is important  for  model interpretability and  reliability. Although such models 
demonstrate exceptional learning capabilities and high prediction accuracy  consistent with system identification principles \cite{deep2025}, their enormous scale and even infinite-dimensional nature bring significant challenges for parameter estimation.

Several studies have addressed parameter estimation for large models with infinite-dimensional parameter spaces. For example, \cite{Sudipta2025} investigated the identification of transfer function coefficients in exponentially stable single-input single-output infinite-dimensional systems, while \cite{Yin2022} proposed an infinite-dimensional sparse learning algorithm using atomic norm regularization for linear system identification. More closely related to the problem discussed in this paper is the work of \cite{Guo1991}, which developed an estimation theory for autoregressive models with exogenous inputs of infinite orders (ARX($\infty$)). Based on  this,
Dai and Guo in \cite{DAI2025} recently established strong consistency of parameter estimates and prediction performance for large  models with saturated output observations.
However,  all the studies mentioned above are limited to a centralized or single-agent estimation setting.

As model sizes expand, centralized  learning or estimation approaches become increasingly impractical due to computational constraints. Furthermore, the storage requirements for large historical datasets make online learning particularly important in large-scale learning tasks. Consequently,
distributed adaptive estimation  has emerged as a promising alternative, where  multiple nodes collaboratively estimate system parameters through localized computations and neighbor-to-neighbor communication. It offers several advantages over centralized methods, such as superior flexibility, stronger fault tolerance for node failures, and  lower communication and computation load. These benefits  promote 
the development of distributed adaptive estimation algorithms  constructed using various algorithmic strategies such as incremental, diffusion, and consensus-based approaches  \cite{7858743,YANG2022105265,Ganauto,Sayed2011Incre,Gan-tnnls,distributedsg2,Schizas2009,Zhangling2017,Mateos2012,LEE2023105505,Nautiyal2022,Philippenko2024}.

A considerable amount of  research has been devoted to investigating the theoretical performance of distributed  estimation algorithms under  certain data conditions. 
Some studies focus on the performance analysis of distributed parameter estimation algorithms using deterministic or time-invariant regression data \cite{YANG2022105265,LEE2023105505}, while some  others focus on the random regression vector case. 
For example,
Barani et al. \cite{distributedsg2} established the convergence of a distributed stochastic gradient descent algorithm employing independent and identically distributed (i.i.d.) signals. Schizas et al. \cite{Schizas2009} provided a stability analysis for a distributed least mean squares (LMS)-type adaptive algorithm, assuming strictly stationary and ergodic regression vectors. Under the assumption of independent regression vectors with positive-definite covariance matrices, Huang et al. \cite{Huang2024} analyzed the mean stability along with the transient and steady-state mean-square performance of Bayesian-learning-based distributed LMS algorithms. Cai et al. \cite{CAI2022} investigated the performance of the $q$-gradient diffusion LMS algorithm using zero-mean, spatially and temporally independent regression vectors. Azarnia \cite{Azarnia2023} studied the steady-state behavior of a deficient-length diffusion LMS algorithm, modeling the regressor components as uncorrelated zero-mean Gaussian processes. Nautiyal et al. \cite{Nautiyal2022} analyzed sparsity-aware diffusion adaptive filters for mean stability, assuming the input signal vectors were spatially and temporally i.i.d. and generated from a zero-mean Gaussian process. Similarly, Lee et al. \cite{math2023} studied the sparse diffusion LMS algorithm with hard thresholding under independent input conditions.
Since recursive least squares (RLS) algorithms often yield more accurate transient estimates and exhibit faster convergence compared to LMS algorithms, distributed RLS frameworks were developed from the perspective of adaptive networks \cite{Cattivelli2008}. Zhang et al. \cite{Zhangling2017} subsequently analyzed the mean-square performance of a diffusion forgetting factor least-squares algorithm, assuming independent input signals. Further, Yu et al. \cite{Yuyi2019} analyzed the mean-square performance of a diffusion RLS algorithm using zero-mean and spatially independent regression vectors. Alireza et al. \cite{Alireza2021} analyzed the mean-square convergence and stability properties of a diffusion RLS algorithm under the assumption of independent input regression vectors possessing a time-invariant covariance matrix.

We note that most existing theoretical results mentioned in the above literature were established by requiring the independence, stationarity or ergodicity conditions  of regression vectors. 
However, for the typical models
such as autoregressive  with exogenous
input (ARX) model and Hammerstein system,  the regressors are inherently constructed from past noisy input and output signals. This internal feedback loop inevitably induces  temporal correlation, potential non-stationarity, and a breakdown of classical independence assumptions in the regression vectors.

To overcome this difficulty, 
recent efforts have focused on relaxing the statistical assumptions required for convergence analysis
\cite{XIE2018,Gan-2024,Gan-tnnls, Wang2021}. 
A particularly noteworthy contribution comes from Xie et al. \cite{Xie2021}, who successfully established the convergence guarantees for a distributed RLS algorithm under a cooperative non-persistent  excitation (non-PE) condition.
Nevertheless, there exists a fundamental limitation  in current theoretical developments of distributed estimation: nearly all
existing theoretical frameworks for distributed adaptive estimation including \cite{7858743,YANG2022105265,Ganauto,Sayed2011Incre,Gan-tnnls,distributedsg2,Schizas2009,Mateos2012,LEE2023105505,XIE2018,Philippenko2024,Gan-2024,Xie2021,Huang2024,CAI2022,Azarnia2023,Nautiyal2022,math2023,Cattivelli2008,Zhangling2017,Yuyi2019, Wang2021, Alireza2021} are inherently designed for and constrained to finite-order or finite-dimensional system models.
Such frameworks are ill-suited for addressing the  challenge presented by large models with infinite parameters.
The unbounded growth in parameter dimensionality requires the further development  of existing mathematical tools to  characterize the convergence behavior of distributed parameter estimates.

This paper investigates  convergence properties of distributed adaptive estimation and prediction  for stochastic  large regression models over time-varying directed multi-agent networks. Each agent aims to collaboratively estimate an infinite-dimensional unknown parameter matrix by sharing and integrating information from neighboring agents.
The primary theoretical challenge lies in analyzing the  properties of the product of asymmetric, non-independent  and non-fixed dimensional random matrices over time-dependent digraphs. We employ double-array martingale theory, stochastic Lyapunov functions and algebraic graph theory to deal with the above issue. 
The main contributions of this work are summarized as follows:
\begin{itemize}
	\item We propose a novel distributed recursive least squares algorithm  capable of handling  stochastic regression models with infinite parameters over time-varying directed graphs. 
	Unlike the traditional diffusion RLS algorithms, our algorithm adapts to regressors whose dimension grows using a non-decreasing sequence, addressing large models like  ARX($\infty$) system. Our algorithm is different from previous distributed estimation algorithms with finite unknown parameters \cite{7858743,YANG2022105265,Ganauto,Sayed2011Incre,Gan-tnnls,distributedsg2,Schizas2009,Mateos2012,LEE2023105505,XIE2018,Gan-2024,Xie2021,Huang2024,CAI2022,Azarnia2023,Nautiyal2022,math2023,Cattivelli2008,Philippenko2024,Zhangling2017,Yuyi2019,Alireza2021}.   
	
	\item We establish almost sure convergence of the proposed algorithm under a quite weak cooperative excitation condition, without requiring independence or stationarity assumptions on regressors commonly used in the existing literature.
	This condition jointly incorporates temporal and spatial network information to capture the cooperative effect of multiple agents compared to \cite{Guo1991}, enabling the  entire multi-agent  system to achieve global estimation, even if any individual agent cannot due to a lack of necessary information. 
	
	\item 
	We derive the  regret analysis for distributed large-model prediction, proving that the averaged regret converges to zero  without resorting to any excitation conditions. This theoretical guarantee confirms the strong predictive performance of the distributed adaptive predictor.  
	The above convergence and prediction results can be reduced to the distributed estimation results for finite-dimensional stochastic regression models in \cite{Xie2021} when the growth rate of regressor's dimension is taken as a  fixed upper bound.

\end{itemize}

The remainder of the paper is organized as follows: Section \ref{sec:formulation} formulates the problem, introduces the system model and communication topology. Section \ref{sec:algorithm} details the distributed recursive least squares algorithm with increasing dimensions. Section \ref{convergence} presents the  assumptions and Key lemmas. Section \ref{main_results} contains the parameter convergence and regret analysis results.  Concluding remarks are given in Section \ref{section_conclusion}.

\section{Problem Formulation}\label{sec:formulation}

\subsection{Matrix theory}
This paper employs  $\mathbb{R}^m$ to denote  the space of $m$-dimensional real vectors, and  $\mathbb{R}^{m\times n}$ to represent real matrices with $m$ rows and $n$ columns. We use $ I_m$ and $\bm 0_m$  to denote the $m$-dimensional identity matrix and $m$-dimensional zero matrix, respectively. For any matrix $ A \in \mathbb{R}^{m \times n}$, its Euclidean norm is defined as $\| A\| \triangleq \left( \lambda_{\max}( A  A^{\top}) \right)^{1/2}$, where the superscript $\top$ denotes matrix transposition and $\lambda_{\max}(\cdot)$ denotes the largest eigenvalue. Correspondingly, $\lambda_{\min}(\cdot)$ represents the smallest eigenvalue and $\operatorname{tr}(\cdot)$ denotes the matrix trace. The operator $\text{col}(\cdot,\dots,\cdot)$ constructs a column vector by vertical stacking of specified elements, while $\text{diag}(\cdot,\dots,\cdot)$ forms a block-diagonal matrix from given vectors or matrices.
The Kronecker product between matrices $ A$ and $ B$ is denoted by $ A \otimes  B$. For symmetric matrices $ X,  Y \in \mathbb{R}^{n \times n}$, the expressions $ X \geq  Y$, $ X >  Y$, $ X \leq  Y$, and $ X <  Y$ indicate that $ X-Y$ is positive semi-definite, positive definite, negative semi-definite, or negative definite, respectively. Regarding matrix sequences $\{ A_t\}$ and positive scalar sequences $\{a_t\}$, the notation $ A_t = O(a_t)$ denotes the existence of a positive constant $C$ independent of $t$ and $a_t$ such that $\|  A_t \| \leq C a_t$ for all $t \geq 0$.
Similarly, $A_k = o(b_k)$ means that $\lim_{k \to \infty} \|A_k\| / b_k = 0$.
\begin{lemma}[Matrix inversion formula \cite{Zielke1968}] \label{wl1}
	For matrices $ A$, $ B$,  $ C$ and  $ D$ of compatible dimensions, the following formula
	$$( A+ B D C)^{-1}= A^{-1}- A^{-1} B( D^{-1}+ C A^{-1} B)^{-1} C  A^{-1}	$$
	holds, provided that the relevant matrices are invertible.
\end{lemma}

\subsection{Graph theory}
A weighted digraph $\mathcal{G}=(\mathcal{I},\mathcal{E},\mathcal{A})$ is a triplet, where $\mathcal{I} = \{1,\dots,n\}$ denotes the vertex set corresponding to agents, $\mathcal{E} \subseteq \mathcal{I} \times \mathcal{I}$ represents the edge set with $(i,j) \in \mathcal{E}$ indicating that agent $i$ can receive information from agent $j$, and $\mathcal{A} \in \mathbb{R}^{n \times n}$ is the weighted adjacency matrix satisfying $a_{ij} > 0$ if and only if $(i,j) \in \mathcal{E}$. 
A path in $\mathcal{G}$ is an ordered sequence of vertices where each consecutive pair is connected by an edge, and the digraph is said to be strongly connected if there exists a path between every pair of distinct vertices. 
For a given positive integer $k$, the union of $k$ digraphs $\{\mathcal{G}_j=(\mathcal{I}, \mathcal{E}_j,\mathcal{A}_j), 1\leq j\leq k\}$ with the same vertex set is denoted by $\cup^{k}_{j=1}\mathcal{G}_j=(\mathcal{V}, \cup^k_{j=1}\mathcal{E}_j,\frac{1}{k}\sum^k_{j=1}\mathcal{A}_j)$.
A sequence of digraphs $\left\{\mathcal{G}_{k}=(\mathcal{I},\mathcal{E}_{k},\mathcal{A}_{k})\right\}_{k\geq 1}$ is $\delta$-nondegenerate for some $\delta>0$, if the weights are uniformly bounded away from zero by $\delta$ whenever positive, i.e., for each $k\geq 1$, $a^{(k)}_{ij}\triangleq\left(\mathcal{A}_{k}\right)_{ij}>\delta$ whenever $a^{(k)}_{ij}>0$.
A sequence $\left\{\mathcal{G}_{k}\right\}_{k\geq 1}$ is $L$-jointly connected  for some positive integer $L$, if for every integer $t \geq 1$, the union digraph $\mathcal{G}_{tL} \cup \cdots \cup \mathcal{G}_{(t+1)L-1}$ is strongly connected.
The weighted out-degree and in-degree of vertex $i \in \mathcal{I}$ are respectively defined as $d_{\text{out}}(i) = \sum_{j=1}^n a_{ij}$ and $d_{\text{in}}(i) = \sum_{j=1}^n a_{ji}$. The digraph is called weight-balanced when $d_{\text{out}}(i) = d_{\text{in}}(i)=1$ holds for all $i \in \mathcal{I}$.

\subsection{Observation model}
This paper considers a  multi-agent network consisting of $n$ agents.  The observation model of each agent $i$ at time $k$ is described by the following  type of stochastic large regression model:
\begin{align}
	y^{\top}_{k+1,i} = \varphi^{\top}_{k,i}{\Theta}+w^{\top}_{k+1,i},~k\geq0,
	\label{model}  
\end{align}
where ${y}_{k,i}\in\mathbb{R}^{m}$ is the observation vector of agent $i$ at time $k$, ${\varphi}_{k,i}\in\mathbb{R}^{\infty}$ is the infinite-dimensional stochastic regression vector of agent $i$ at time $k$ with countably infinite components, ${w}_{k+1,i}\in\mathbb{R}^{m}$ is the noise process, ${\Theta}\in\mathbb{R}^{\infty\times m}$  is the unknown parameter matrix
to be estimated and satisfy the following summability condition:
\begin{gather}
	\sum^{\infty}_{q=1}\|\Theta^{[q]}\|<\infty,\label{con}
\end{gather}
where $\Theta^{[q]}\in\mathbb{R}^{m}$ is defined as the $q$-th row of the matrix $\Theta$.
Note that model (\ref{model})  naturally include the finite-dimensional case, since condition (\ref{con}) is automatically satisfied when  $\Theta^{[q]}$
vanishes for large $q$.
We assume that the network communication topology of $n$ agents is time-dependent and described by a
sequence of  digraphs $\left\{\mathcal{G}_{k}=(\mathcal{I},\mathcal{E}_{k},\mathcal{A}_{k})\right\}_{k\geq 1}$.
The neighbor set of agent $i$ at time $k$ is denoted as $\mathcal{N}_{i,k}=\{j\in \mathcal{I}, (i,j)\in\mathcal{E}_k \}$, and the agent $i$ is also included in this set.

Model (\ref{model}) incorporates a broad range of practical
scenarios, including   models with varying or expanding parameters
dimensions. Here are two  typical examples that can be  incorporated into model (\ref{model}).
\begin{example} \label{ex1}
	Consider the well-known ARX model with infinite orders:
	\begin{align}
		y_{k+1,i}&=\sum^{\infty}_{q=1}(A_q y_{k+1-q,i}+B_qu_{k+1-q,i})+w_{k+1,i}, ~~k\geq0;\nonumber\\
		y_{k,i}&=w_{k,i}=0, u_{k,i}=0,~ \forall k<0,\label{arx}
	\end{align}
	where   ${u}_{k,i}\in\mathbb{R}^{l}$ denote the $l$-dimensional  input of agent $i$ at time instant $k$,
	$A_{k}\in\mathbb{R}^{m\times m}, B_{k}\in\mathbb{R}^{m\times l} (k=1, 2,3\dots )$ are the unknown matrices  and satisfy 
	$
	\sum^{\infty}_{k=1}(\|A_k \|+\|B_k \|)<\infty.
	$
	Set $\Theta^{\top}=[A_1,B_1,A_2,B_2,\cdots]\in\mathbb{R}^{m\times \infty}$ and $\varphi_{k,i}=(y^\top_{k,i},u^\top_{k,i},y^\top_{k-1,i},u^\top_{k-1,i},\cdots)^\top\in\mathbb{R}^\infty$. Then 
	the above  ARX($\infty$) model (\ref{arx}) can be transformed into the stochastic large  model (\ref{model}).
\end{example}

\begin{example} (See Example 2 in \cite{DAI2025}) 
	We consider  a  neural network model with a single hidden layer for agent $i$: $$y_{k+1,i}=\sum^{\infty}_{q=1}B_qu_{k,i,q}+w_{k+1,i},~~ k\geq 0,$$ where the hidden layer outputs  $u_{k,i,q} = \sigma(A_q^\top \phi_{k,i} + C_q)$ with $\sigma(\cdot)$ being the activation function and $\phi_{k,i}$ denotes the input feature vector at time step $k$ of agent $i$. In practical adaptive control applications, the parameter estimation process typically involves an initial offline phase where reinforcement learning techniques provide preliminary approximations for parameter set $\{A_q, B_q, C_q\}_{q \geq 1}$, followed by an online  phase where $\{A_q, C_q\}_{q \geq 1}$ remain fixed to preserve computational resources while only the output weights $\{B_q\}_{q \geq 1}$ receive continuous updates. This strategy ensures that the  signals $u_{k,i,q}$, which depend exclusively on the fixed parameters estimates for $\{A_q, C_q\}_{q \geq 1}$ and the current input $\phi_{k,i}$, can be treated as known quantities for online estimation and prediction.
\end{example}

The goal of this paper is to design a distributed adaptive identification algorithm that enables all agents to cooperatively estimate the unknown infinite parameters $\Theta\in\mathbb{R}^{\infty\times m}$ by  the local  process $\{y_{k,j},\varphi_{k,j},k\geq 1\}_{j\in\mathcal{N}_{i,k}}$ over time-varying communication topology.

\section{Distributed estimation algorithm}\label{sec:algorithm}
Let $\{p_t\}$ be any non-decreasing sequence of positive integers and $p_t=O(t)$. 
Define the $p_t$-dimensional regressor vector as
\begin{align}
	\varphi_{k,i}(t)=\text{col}\left\{\varphi_{k,i}^{[1]},\varphi_{k,i}^{[2]}\cdots, \varphi_{k,i}^{[p_t]}\right\}\in\mathbb{R}^{p_t},\label{varphi}
\end{align}
where $\varphi_{k,i}^{[q]}~ (q=1,2,...,p_t)$ represents the $q$-th element of the vector $\varphi_{k,i}$.

In order to construct the distributed estimation algorithm,  
We first introduce the following  recursive estimation error function  of each sensor $i$ with the parameter $\Gamma\in\mathbb{R}^{p_t\times m}$:
\begin{align}
	\!\!\!\!\sigma_{k+1,i}(\Gamma)\!=\!\!\!\sum_{j\in\mathcal{N}_{i,k}}a^{(k)}_{ij}\left(\sigma_{k,j}(\Gamma)\!+\!\|y_{k+1,j}\!-\!\Gamma^{\top}\varphi_{k,j}(t)\|^2\right)\label{sigma1}
\end{align}
where  
$ \sigma_{0,i}(\Gamma)=\frac{1}{\beta}\operatorname{tr}\left(\Gamma^{\top}\Gamma\right)$ with $\beta>0$. Note that the dimension of  $\varphi_{k,j}(t)$  defined in (\ref{varphi}) is increasing with $p_t$.
By minimizing the above local cost function  $\sigma_{k+1,i}(\Gamma)$,  the estimate of each sensor $i$ can be obtained as 
\begin{align}
	\Theta_{k+1,i}(t)\triangleq \text{argmin}_{\Gamma\in\mathbb{R}^{p_t\times m}}\sigma_{k+1,i}(\Gamma).\label{re3}
\end{align}
We derive the recursive form of $\Theta_{k+1,i}(t)\in\mathbb{R}^{p_t\times m}$ below.

Let $\sigma_{k}(\Gamma)=\text{col}\{\sigma_{k,1}(\Gamma),\cdots,\sigma_{k,n}(\Gamma)\} $ and $
e_{k+1}(\Gamma)$$=\text{col}\left\{\|y_{k+1,1}\!-\!\Gamma^{\top}\varphi_{k,1}(t)\|^2,\cdots,\|y_{k+1,n}\!-\!\Gamma^{\top}\varphi_{k,n}(t)\|^2\right\}.$ We can recursively get from (\ref{sigma1}) that
\begin{align}
	&\sigma_{k+1}(\Gamma)=\mathcal{A}_k\sigma_{k}(\Gamma)+\mathcal{A}_ke_{k+1}(\Gamma)\nonumber\\
	&=	\mathcal{A}_k\mathcal{A}_{k-1}\sigma_{k-1}(\Gamma)
	+\mathcal{A}_k\mathcal{A}_{k-1}e_{k}(\Gamma)+\mathcal{A}_ke_{k+1}(\Gamma)\nonumber\\
	&=\cdots=\mathcal{A}_k\mathcal{A}_{k-1}\cdots\mathcal{A}_{0}\sigma_{0}(\Gamma)+
	\sum^k_{l=0}\mathcal{A}_k\cdots\mathcal{A}_{l}e_{l+1}(\Gamma),\label{re1}
\end{align}
where $\mathcal{A}_k$ is the weighted adjacency matrix at time $k$.

For convenience, denote $\mathcal{A}(k,l)=\mathcal{A}_k\cdots\mathcal{A}_{l}$ and $a^{(k,l)}_{ij}$ as the $i$-th row and the $j$-th column element of $\mathcal{A}(k,l)$.
Since the weighted adjacency matrix is stochastic,  we have from (\ref{re1})
\begin{align*}
	\sigma_{k+1,i}(\Gamma)=\frac{1}{\beta}\operatorname{tr}\left(\Gamma^{\top}\Gamma\right)\nonumber
	+\sum^n_{j=1}\sum^{k}_{l=0}a^{(k,l)}_{ij}
	\|y_{l+1,j}-{\Gamma}^{\top}\varphi_{l,j}(t)\|^2.
\end{align*}
Substituting the above equation into (\ref{re3}) and by the matrix derivative rules, we can obtain that
\begin{align}
	\Theta_{k+1,i}(t)= P_{k+1,i}(t)\left(\sum^n_{j= 1}\sum^k_{l= 0}a^{(k,l)}_{ij} \varphi_{l,j}(t)y^{\top}_{l+1,j}\right),\label{theta}
\end{align}
where
\begin{align}
	\!\!P_{k+1,i}(t)=\!\Bigg(\!\sum^n_{j=1}\sum^k_{l=0}a^{(k,l)}_{ij}\varphi_{l,j}(t)\varphi^{\top}_{l,j}(t)
	\!+\frac{1}{\beta}I_{p_t}\!\Bigg)^{\!-1}\!\!\label{p}.
\end{align}
Then it can be  equivalently written as the following recursive form,
\begin{align}
	P^{-1}_{k+1,i}(t)=\sum_{j\in\mathcal{N}_{i,k}}a^{(k)}_{ij}(P^{-1}_{k,j}(t)+\varphi_{k,j}(t)\varphi^{\top}_{k,j}(t))\label{P}
\end{align}
with $	P^{-1}_{0,i}(t)=\frac{1}{\beta}I_{p_t}$. Using a similar recursive relationship, equation (\ref{theta}) can be transformed into the following expression:
\begin{align}
	\Theta_{k+1,i}(t)=&P_{k+1,i}(t)\sum_{j\in\mathcal{N}_{i,k}}a^{(k)}_{ij}\Big(P^{-1}_{k,j}(t)\Theta_{k,j}(t)\nonumber\\
	&+\varphi_{k,j}(t)y
	^{\top}_{k+1,j}\Big). \label{211}
\end{align} 
To facilitate further derivation, denote 
\begin{align}
	\bar P^{-1}_{k+1,j}(t)\triangleq P^{-1}_{k,j}(t)+\varphi_{k,j}(t)\varphi^{\top}_{k,j}(t).\label{bar}
\end{align}
Applying the matrix inversion formula (i.e., Lemma \ref{wl1}), we derive that 
\begin{align}
	{\bar P}_{k+1,i}(t)&= P_{k,i}(t)- \frac{P_{k,i}(t)\varphi_{k,i}(t)\varphi^{\top}_{k,i}(t) P_{k,i}(t)}{1+\varphi^{\top}_{k,i}(t)P_{k,i}(t)\varphi_{k,i}(t)}\label{12}
\end{align} 
Then by (\ref{P}), it is clear that 
\begin{align}
	P^{-1}_{k+1,i}(t)=\sum_{j\in\mathcal{N}_{i,k}}a^{(k)}_{ij}\bar P^{-1}_{k,j}(t).
\end{align}
Thus by (\ref{12}), we have
\begin{align}
	{\bar\Theta}_{k+1,i}(t)&\triangleq {\bar P}_{k+1,i}(t)\left(P^{-1}_{k,j}(t)\Theta_{k,j}(t)
	+\varphi_{k,j}(t)y
	^{\top}_{k+1,j}\right)\nonumber\\
	&={\Theta}_{k,i}(t)+	\frac{P_{k,i}(t)\varphi_{k,i}(t)}{1+\varphi^{\top}_{k,i}(t)P_{k,i}(t)\varphi_{k,i}(t)}\nonumber\\
	&\hspace*{0.3in} \left(y^{\top}_{k+1,i}-\varphi^{\top}_{k,i}(t){\Theta}_{k,i}(t)\right).\label{adap}
\end{align}
By (\ref{211}), we have 
\begin{align}
	{\Theta}_{k+1,i}(t)= P_{k+1,i}(t)\sum_{j\in \mathcal{N}_{i,k}}a^{(k)}_{ij}{\bar P}^{-1}_{k+1,j}(t){\bar\Theta}_{k+1,j}(t).
\end{align}
Therefore, we get the following diffusion-type  recursive least  squares (RLS) algorithm with increasing dimensions, see  Algorithm  \ref{algorithm}.          
\begin{algorithm}[!htb]
	\caption{Distributed RLS  with increasing dimensions \label{algorithm}}
	\begin{algorithmic}[0]   
		\State $\mathbf{Initialization.}$ For each agent $i\in\{1,\cdots,n\}$ and any $t>0$,  define  $p_t$-dimensional regressors $\varphi_{k,i}(t)$ by (\ref{varphi}),
		and begin with arbitrary  initial matrix ${\Theta}_{0,i}(t)\in\mathbb{R}^{p_t\times m}$ and 
		initial matrix $P_{0,i}(t)=\beta I_{p_t}\in\mathbb{R}^{p_t\times p_t}$ with $\beta>0$ being choosing arbitrarily.
		\For{ \rm{each time} $k=0,1,\cdots,t-1$}
		\For{ \rm{each~ agent~} $i=1,\cdots,n$}
		\State $\mathbf{Step\ 1.}$ Generate ${\bar\Theta}_{k+1,i}(t)$ and ${\bar P}_{k+1,i}(t)$ based  on $y_{k+1,i},{\Theta}_{k,i}(t)$, $ P_{k,i}(t)$, $\varphi_{k,i}(t)$:
		\begin{align*}
			&b_{k,i}(t)=(1+\varphi^{\top}_{k,i}(t)P_{k,i}(t)\varphi_{k,i}(t))^{-1}\\
			&{\bar\Theta}_{k+1,i}(t)={\Theta}_{k,i}(t)+	b_{k,i}(t)P_{k,i}(t)\varphi_{k,i}(t)\nonumber\\
			&\hspace*{0.8in} (y^{\top}_{k+1,i}-\varphi^{\top}_{k,i}(t){\Theta}_{k,i}(t)),\\
			&{\bar P}_{k+1,i}(t)= P_{k,i}(t)- b_{k,i}(t)P_{k,i}(t)\varphi_{k,i}(t)\varphi^{\top}_{k,i}(t) P_{k,i}(t),
		\end{align*}
		
		\State $\mathbf{Step\ 2.}$ Generate $ P^{-1}_{k+1,i}(t)$ and ${\Theta}_{k+1,i}(t)$  by a convex combination of ${\bar\Theta}_{k+1,j}(t)$ and ${\bar P}^{-1}_{k+1,j}(t)$:
		\begin{align*}
			P^{-1}_{k+1,i}(t)&=\sum_{j\in\mathcal{N}_{i,k}}a^{(k)}_{ij}{\bar P}^{-1}_{k+1,j}(t),\\
			{\Theta}_{k+1,i}(t)&= P_{k+1,i}(t)\sum_{j\in \mathcal{N}_{i,k}}a^{(k)}_{ij}{\bar P}^{-1}_{k+1,j}(t){\bar\Theta}_{k+1,j}(t).
		\end{align*}
		\EndFor
		\EndFor
		\State $\mathbf{Output.}$ ${\Theta}_{t,i}(t).$
	\end{algorithmic}
\end{algorithm}

For convenience, we need to introduce the following 
notations:
\begin{align*}
	&{Y}_{k+1}^{\top}\triangleq{\rm col}\{{y}_{k+1,1}^{\top},\dots,{y}_{k+1,n}^{\top}\}, &(n\times m)\\
	&{\Phi}_{k}^{\top}(t)\triangleq{\rm diag}\{{\varphi}_{k,1}^{\top}(t),\dots,{\varphi}_{k,n}^{\top}(t)\}, &(n\times np_t)\\
	&{W}_{k+1}^{\top}\triangleq{\rm col}\{{w}_{k+1,1}^{\top},\dots,{w}_{k+1,n}^{\top}\}, &(n\times m)\\
	&\Theta(t)\triangleq{\rm col}\{\Theta^{[1]},\cdots,\Theta^{[p_t]}\}, &(p_t\times m)\\
	&{\Gamma}(t)\triangleq{\rm col}\{\underbrace{{\Theta(t)},\dots,{\Theta(t)}}_{n}\}, &(np_t\times m)\\
	&{\Theta}_{k}(t)\triangleq{\rm col}\{{\Theta}_{k,1}(t),\dots,{\Theta}_{k,n}(t)\}, &(np_t\times m)\\
	&\bar{{\Theta}}_{k}(t)\triangleq{\rm col}\{\bar{{\Theta}}_{k,1}(t),\dots,\bar{{\Theta}}_{k,n}(t)\}, &(np_t\times m)\\
	&\widetilde{{\Theta}}_{k}(t)\triangleq{\rm col}\{\widetilde{{\Theta}}_{k,1}(t),\dots,\widetilde{{\Theta}}_{k,n}(t)\}, &(np_t\times m)\\
	&~~~~~~~\text{where}~ \widetilde{{\Theta}}_{k,i}(t)={\Theta(t)}-{\Theta}_{k,i}(t),\\
	&\widetilde{\bar{{\Theta}}}_{k}(t)\triangleq{\rm col}\{\widetilde{\bar{{\Theta}}}_{k,1}(t),\dots,\widetilde{\bar{{\Theta}}}_{k,n}(t)\}, &(np_t\times m)\\
	&~~~~~~~\text{where}~ \widetilde{\bar{{\Theta}}}_{k,i}(t)={\Theta(t)}-\bar{{\Theta}}_{k,i}(t),\\
	&{b}_{k}(t)\triangleq{\rm diag}\{b_{k,1}(t),\dots,b_{k,n}(t)\}, &(n\times n)\\
	&{c}_{k}(t)\triangleq{b}_{k}(t)\otimes I_{p_t}, &(np_t\times np_t)\\
	&{P}_{k}(t)\triangleq{\rm diag}\{P_{k,1}(t),\dots,P_{k,n}(t)\}, &(np_t\times np_t)\\
	&\bar{{P}}_{k}(t)\triangleq{\rm diag}\{\bar{P}_{k,1}(t),\dots,\bar{P}_{k,n}(t)\}, &(np_t\times np_t)\\
	&{\mathscr{A}_k}(t)\triangleq\mathcal{A}_k\otimes I_{p_t}, &(np_t\times np_t)\\
	&\epsilon^{\top}_{k,i}(t)=\sum^{\infty}_{q=p_t+1}\varphi^{[q]}_{k,i}\Theta^{[q]}, &(1\times m)\\
	&\epsilon^{\top}_k(t)=\text{col}\{\epsilon^{\top}_{k,1}(t),\cdots,\epsilon^{\top}_{k,n}(t)\}. &(n\times m)\\
\end{align*}
By the definition of $\epsilon^{\top}_{k,i}(t)$ and (\ref{model}), we have 
\begin{align}
	{y}_{k+1,i}^{\top}=\varphi^{\top}_{k,i}(t){\Theta(t)}+w^{\top}_{k+1,i}+\epsilon^{\top}_{k,i}(t).\label{y1}
\end{align}
Thus we obtain the following compact form:
\begin{gather}
	{Y}_{k+1}^{\top}=\Phi^{\top}_{k}(t){\Gamma(t)}+W^{\top}_{k+1}+\epsilon^{\top}_{k}(t).\label{1}
\end{gather}
By Step 1 in Algorithm \ref{algorithm}, we have
\begin{align}\label{al1}
	\left\{
	\begin{aligned}	&b_k(t)=(I_{n}+\Phi^{\top}_{k}(t)P_{k}(t)\Phi_{k}(t))^{-1},\\
		&\bar{{\Theta}}_{k+1}(t)={\Theta}_{k}(t)+c_k(t){P}_{k}(t)\Phi_{k}(t)({Y}_{k+1}^{\top}-\Phi^{\top}_{k}(t){\Theta_k(t)}),\\
		&\bar{{P}}_{k+1}(t)={P}_{k}(t)-c_k(t){P}_{k}(t)\Phi_{k}(t)\Phi^{\top}_{k}(t)P_{k}(t).
	\end{aligned}
	\right.
\end{align}
Substituting (\ref{1}) into (\ref{al1}), we can obtain the following error equation of Algorithm \ref{algorithm}.
\begin{align}
	\widetilde{\bar{{\Theta}}}_{k+1}(t)=&\Gamma(t)-\bar{{\Theta}}_{k+1}(t)\nonumber\\
	=&\Gamma(t)-{\Theta}_{k}(t)-c_k(t){P}_{k}(t)\Phi_{k}(t)\Big(\Phi^{\top}_{k}(t){\Gamma(t)}\nonumber\\
	&+W^{\top}_{k+1}+\epsilon^{\top}_{k}(t)-\Phi^{\top}_{k}(t){\Theta_k(t)}\Big)\nonumber\\
	=&\left(I_{np_t}-c_k(t){P}_{k}(t)\Phi_{k}(t)\Phi^{\top}_{k}(t)\right)\widetilde{{\Theta}}_{k}(t)\nonumber\\
	&-c_k(t){P}_{k}(t)\Phi_{k}(t)\left(W^{\top}_{k+1}+\epsilon^{\top}_{k}(t)\right)\nonumber\\
	=&\bar{{P}}_{k+1}(t){P}^{-1}_{k}(t)\widetilde{{\Theta}}_{k}(t)\nonumber\\
	&-c_k(t){P}_{k}(t)\Phi_{k}(t)\left(W^{\top}_{k+1}+\epsilon^{\top}_{k}(t)\right).\label{bartheta}
\end{align}
By Step 2 in Algorithm \ref{algorithm}, for each agent $i\in\{1,2,...,n\}$
$$
P_{k+1,i}(t)\sum_{j\in \mathcal{N}_{i,k}}a^{(k)}_{ij}{\bar P}^{-1}_{k+1,j}(t){\Theta}(t)={\Theta}(t).
$$
Thus we obtain that 
\begin{align}
	&P_{k+1}(t){\mathscr{A}_k}{\bar P}^{-1}_{k+1}(t)\Gamma(t)\nonumber\\
	=&\text{col}\Big\{P_{k+1,1}(t)\sum_{j\in \mathcal{N}_{1,k}}a^{(k)}_{1j}{\bar P}^{-1}_{k+1,j}(t){\Theta}(t),\cdots,\nonumber\\
	&~~~~~P_{k+1,n}(t)\sum_{j\in \mathcal{N}_{n,k}}a^{(k)}_{nj}{\bar P}^{-1}_{k+1,j}(t){\Theta}(t)\Big\}=\Gamma(t).
\end{align}
Hence by Step 2 in Algorithm \ref{algorithm} and (\ref{bartheta}), we have
\begin{align}
	&\widetilde{\Theta}_{k+1}(t)\nonumber\\
	=&\Gamma(t)-\Theta_{k+1}(t)\nonumber\\
	=&
	\Gamma(t)-P_{k+1}(t){\mathscr{A}_k}(t){\bar P}^{-1}_{k+1}(t)\bar{{\Theta}}_{k+1}(t)\nonumber\\
	=&P_{k+1}(t){\mathscr{A}_k}(t){\bar P}^{-1}_{k+1}(t)\widetilde{\bar{{\Theta}}}_{k+1}(t)\nonumber\\
	=&P_{k+1}(t){\mathscr{A}_k}(t){P}^{-1}_{k}(t)\widetilde{{\Theta}}_{k}(t)-P_{k+1}(t){\mathscr{A}_k}(t){\bar P}^{-1}_{k+1}(t)\nonumber\\
	&\cdot c_k(t){P}_{k}(t)\Phi_{k}(t)\left(W^{\top}_{k+1}+\epsilon^{\top}_{k}(t)\right).\label{error_equation}
\end{align}

Compared with the traditional distributed least squares algorithm designed for stochastic systems with finite parameters (cf. \cite{Xie2021}), this paper focuses on a large regression model with infinite-dimensional regressors. We use a time-varying dimensional space parameterized by $p_t$ to approximate  the infinite-dimensional regressor space.
Consequently, the  variables  in Algorithm \ref{algorithm} and error equation (\ref{error_equation})  exhibit  time-varying increasing dimensions,  which bring us big challenges for rigorous convergence analysis.

\section{Preliminary}\label{convergence}
To establish our main convergence results, we first introduce basic assumptions,  followed by several key lemmas for our theoretical analysis.
\subsection{Assumptions}
\begin{assumption}\label{a1}
	(i)  All digraphs $\{\mathcal{G}_{k}\}$ are weight-balanced.
	(ii) The  sequence 	of digraphs $\{\mathcal{G}_{k}\}_{k\geq 1}$
	is $\delta$-nondegenerate and $L$-jointly connected for some positive integer $L$. 
\end{assumption}
\begin{remark} \label{remark1}
	The parameter $\delta$ in the $\delta$-nondegeneracy condition of the directed graph sequence serves as a uniform lower bound on the weakest connection in the sequence, ensuring that no path's weight decays too rapidly in the weighted adjacency matrices. This property is important for  analyzing products of graph adjacency matrices over time. Assumption \ref{a1} guarantees that for any $k>s$, the entry  $a^{(k,s)}_{ij}$  in the weighted adjacency matrix product satisfies the following properties  (see Lemma 1  in \cite{Nedic2009}):
	
	(a) $a^{(k,s)}_{ii}\geq \delta^{k-s+1}$ for all $i,k,$ and $s$ with $k\geq s$ ;
	
	(b)  $a^{(k,s)}_{ij}\geq \delta^{k-s+1}$ for all $k, s$ with $k\geq s$  and for all $(i,j)\in\mathcal{E}_s\cup\cdots\cup\mathcal{E}_k$;
	
	(c)		 $a^{(k,s)}_{ij}\geq \delta^{k-s+1}$ whenever there exists a vertex $v$  satisfying $(v,j)\in\mathcal{E}_s\cup\cdots\cup\mathcal{E}_r$ for some $r>s$ and $(i,v)\in\mathcal{E}_{r+1}\cup\cdots\cup\mathcal{E}_k$ for $k>r$.
\end{remark}
\begin{assumption}\label{a2}
	For any $i\in\{1,\cdots,n\}$, each element of $\varphi_{k,i}$  is  $\mathscr{F}_k$-measurable, and the noise sequence $\{w_{k,i},\mathscr{F}_k\}$ is a martingale difference sequence  satisfying
	\begin{align*}
		\sup_{k\geq 0}\mathbb{E}[\|w_{k+1,i}\|^{2}|\mathscr{F}_k]<\infty
	\end{align*}
	and  $\|w_{k,i}\|^2=o(d_i(k)),~{\rm a.s.,}$
	where  $\{\mathscr{F}_k\}$ is a non-decreasing family of $\sigma$-algebras  and $\{d_i(k)\}_{k\geq 1}$ is a positive, deterministic, nondecreasing sequence and satisfies
	\begin{align*}
		\sup_k d_i(e^{k+1})/d_i(e^k)<\infty, ~\forall ~i.
	\end{align*}
\end{assumption}  
\begin{remark}
	The condition ``$\|w_{t,i}\|^2=o(d_i(t))$" in Assumption \ref{a2} governs the noise growth rate. This implies that the double array martingale estimation theory (Lemmas \ref{lemma1} and \ref{lemma2} below) is applicable for controlling the aggregate noise impact, expressed as 
	$\max_{1\leq m\leq p_t}\left\|\sum^t_{k=1}f_k(m)w_{k+1}\right\|$. It is straightforward to verify that commonly used bounded or  white Gaussian noises can satisfy this condition. 
\end{remark}
\subsection{Key Lemmas}




\begin{lemma}\label{lemma3}
For any given $l\geq 1$, denote the $p_l$-dimension regression vector of agent $i$ at time $k$ as $\varphi_{k,i}(l)=\text{col}\left\{\varphi_{k,i}^{[1]},\varphi_{k,i}^{[2]}\cdots, \varphi_{k,i}^{[p_l]}\right\}$.  Under Assumptions \ref{a1}(i),  we have  the following property for Algorithm \ref{algorithm} with $\varphi_{k,i}(l)\in\mathbb{R}^{p_l}$:
\begin{align}
	&\operatorname{tr}(V_{t}(l))+\frac{1}{2}\sum^{t-1}_{k=0}\operatorname{tr}\left(\widetilde{\Theta}^{\top}_{k}(l)\Phi_{k}(l)b_k(l)\Phi^{\top}_{k}(l)\widetilde{{\Theta}}_{k}(l)
	\right)\nonumber\\
	\leq & \operatorname{tr}(V_0(l))+4\sum^n_{i=1}\!\sum^{t-1}_{k=0}\|\epsilon_{k,i}(l)\|^2\nonumber\\
	&-2\sum^n_{i=1}\sum^{t-1}_{k=0}b_{k,i}(l)\varphi^{\top}_{k,i}(l)\widetilde{\Theta}_{k,i}(l)w_{k+1,i}\nonumber\\
	&+2\sum^{t-1}_{k=0}\operatorname{tr}\Big(W_{k+1}
	b_k(l)\Phi^{\top}_{k}(l){P}_{k}(l)\Phi_{k}(l)W^{\top}_{k+1}
	\Big),\label{new6}
\end{align}
where $$V_k(l)=\widetilde{\Theta}^{\top}_{k}(l){P}_{k}^{-1}(l)\widetilde{\Theta}_{k}(l),~\epsilon^{\top}_{k,i}(l)=\sum^{\infty}_{q=p_l+1}\varphi^{[q]}_{k,i}\Theta^{[q]}.$$
\end{lemma}

The proof of Lemma \ref{lemma3} is provided in Appendix \ref{proof:Lemma4.2}.
To deal with the accumulated noise effect which takes the form as $\max_{1\leq m\leq p_t}\left\|\sum^t_{k=1}f_k(m)w_{k+1}\right\|$, we introduce the following two double-array martingale lemmas from \cite{huang1990}.
\begin{lemma}\label{lemma1}
{\rm \cite{huang1990}}
Let $\{v_t,\mathscr{F}_t\}$ be an $s$-dimensional martingale difference sequence satisfying
$\|v_t\|=o(\eta(t))~~{\rm a.s.}$,
where the properties of $\eta(t)$ are described as same as $d_i(t)$ in Assumption \ref{a2}. Assume that $f_t(m), t, m=1,2,...,$ is an $\mathscr{F}_t$-measurable, $r\times s$-dimensional random matrix satisfying
\begin{align}
	\|f_t(m)\|\leq C<\infty ~{\rm a.s.} \label{bou}
\end{align}
for all $t$, $m$ and some deterministic constant $C$. 
Then for $p_t=O(\log^{\alpha} t)$ ~$(\alpha>0)$, the following properties hold as $t\rightarrow\infty$,
\begin{align}
	(i)~~~~~	&\max_{1\leq m\leq p_t}\max_{1\leq j\leq t}\left\|\sum^j_{k=1}f_k(m)v_{k+1}\right\|\\
	=&O\left(\max_{1\leq m\leq p_t}\sum^t_{k=1}\|f_k(m)\|^2\right)+o(\eta(t)\log\log t),~~{\rm a.s.}
\end{align}
provided that 
\begin{align}
	\sup_t \mathbb{E}(\|v_{t+1}\|^2|\mathscr{F}_t)<\infty ~~{\rm a.s..}\label{new1}
\end{align}
\begin{align*}
	(ii)~~~~~	&\max_{1\leq m\leq p_t}\max_{1\leq j\leq t}\left\|\sum^j_{k=1}f_k(m)v_{k+1}\right\|\\
	=&O\left(\max_{1\leq m\leq p_t}\sum^t_{k=1}\|f_k(m)\|\right)+o(\eta(t)\log\log t),~~{\rm a.s.}
\end{align*}
provided that 
\begin{align}
	\sup_t \mathbb{E}(\|v_{t+1}\||\mathscr{F}_t)<\infty ~~{\rm a.s..}
\end{align}
\end{lemma}

\begin{lemma}\label{lemma2}	{\rm \cite{huang1990}}
Under the conditions of Lemma \ref{lemma1} except (\ref{bou}). If (\ref{new1}) holds, then
\begin{align*}
	&\max_{1\leq m\leq p_t}\max_{1\leq j\leq t}\left\|\sum^j_{k=1}f_k(m)v_{k+1}\right\|\\
	=&O\left(a_t\log a_t\right)+o(a_t\eta(t)\log\log t)\\
	=&O(1)+o(a^2_t)+o([\eta(t)\log\log t]^2)
	~~{\rm a.s.}
\end{align*}
where $a_t=\max_{1\leq m\leq p_t}g_{t}(m),$ with  $g_0(m)=1$, $g_{t}(m)=\left(1+\sum^t_{k=1}\|f_k(m)\|^2\right)^{\frac{1}{2}}$.
\end{lemma}

Thus from Lemmas \ref{lemma3}-\ref{lemma2}, we have the following results.
\begin{lemma}\label{th1}
Consider the model (\ref{model})-(\ref{con}) and Algorithm \ref{algorithm} with $p_t=O(\log^{\alpha}t)$, $\alpha>0$. Under Assumptions \ref{a1}(i) and \ref{a2}, we have the following property for Algorithm \ref{algorithm}:
\begin{align}
	&\operatorname{tr}\left(\widetilde{\Theta}^{\top}_{t}(t){P}_{t}^{-1}(t)\widetilde{\Theta}_{t}(t)\right)\nonumber\\
	&+\left(\frac{1}{2}+o(1)\right)\sum^{t-1}_{k=0}b_{k,i}(t)\|\varphi^{\top}_{k,i}(t)\widetilde{\Theta}_{k,i}(t)\|^2\nonumber\\
	=&o([d(t)\log\log t]^2)+O\left(s_t\right)\nonumber\\
	&+O\left(p_t\log^+\{\lambda_{\max}({P}_{t}^{-1}(t))\}\right), ~{\rm a.s.,}\label{2}
\end{align}
where   $d(t)=\left(\sum^n_{i=1}d^2_i(t)\right)^{\frac{1}{2}}$ with $d_i(t)$ being defined in Assumption \ref{a2}, $s_t=\sum^n_{i=1}\!\sum^{t-1}_{k=0}\|\epsilon_{k,i}(t)\|^2$ and $\log^+\{x\}=\max\{\log x,0\}$.
	\end{lemma}
	
	\begin{remark}\label{re4.2}
The proof of Lemma \ref{th1} is given in Appendix \ref{proof:th1}. 
From \cite{huang1990}, we know that if $p_t$ is  assumed to satisfy $p_t=O(t^\alpha)$, $\alpha>0$, then the results \eqref{2} in Lemma \ref{th1} still hold, provided that $\log\log t$ in \eqref{2} is replaced by $\log t$.
\end{remark}

\section{Main Results}\label{main_results}
In this section, we will establish the strong consistency of Algorithm \ref{algorithm} under a cooperative excitation condition for the regression signals, and will analyze  the prediction error without relying on any excitation conditions.
\subsection{Parameter convergence}
We first provide a theoretical upper bound for the parameter estimation error induced by Algorithm \ref{algorithm}.

\begin{theorem}\label{th2}
Consider the model (\ref{model})-(\ref{con}) and Algorithm \ref{algorithm} with $p_t=O(\log^{\alpha}t)$, $\alpha>0$. Under Assumptions \ref{a1} and \ref{a2}, then as $t\rightarrow\infty$, the upper bound of estimation error of Algorithm \ref{algorithm} is
\begin{align}
	&\|{\hat\Theta}_{i}(t)-\Theta\|^2\nonumber\\
	=&O\!\left(\frac{1}{\lambda_{\min}(t)}\left[p_t\log r_t\!+\!s_t\! +\!o\!\left([d_t\log\log(t)]^2\right)\right]+\gamma_t\right) \text{a.s.,}\label{th}
\end{align}
where $s(t)$ and $d(t)$ are defined in Lemma \ref{th1} and
\begin{align*}
	&{\hat\Theta}_{i}(t)\triangleq\text{col}\{\Theta_{t,i}(t),\bm 0_m,\bm 0_m,...\}\in\mathbb{R}^{\infty\times m},\\
	&\lambda_{\min}(t)\triangleq\lambda_{\min}\left(\sum^n_{i=1}\sum^{t-nL-2}_{k=0}\varphi_{k,i}(t)\varphi^{\top}_{k,i}(t)+\frac{1}{\beta}I_{p_t}\right),\\
	&r_t\triangleq1+\sum^n_{j=1}\sum^{t-1}_{k=0}\|\varphi_{k,j}(t)\|^2,~
	\gamma_t\triangleq\left(\sum^{\infty}_{q=p_t+1}\left\|\Theta^{[q]}\right\|\right)^2
\end{align*}
with $\Theta^{[q]}\in\mathbb{R}^{m}$ being  the $q$-th row of the matrix $\Theta$.
	\end{theorem}
	
	\begin{proof}
For any given  $t$, there exists a positive integer $K_{t-1}=[ \frac{t-1}{L} ]\footnote{The symbol $[x]$ denotes the maximum integer less than $x$.}$ such that $K_{t-1}L\leq t-1 \leq (K_{t-1}+1)L-1$. Then
by Assumption \ref{a1} (ii), 	for any pair of agents $v, u\in\{1,...,n\}$, there exists a path
$$v_n\triangleq u\rightarrow v_{n-1}\rightarrow v_{n-2}\rightarrow\cdots\rightarrow v_{1}\rightarrow v\triangleq v_0$$
such that   $v_{s}\rightarrow v_{s-1}$ is an directed edge in graph 
$\mathcal{G}_{(K_{t-1}-s) L}\cup\cdots\cup\mathcal{G}_{(K_{t-1}-s+1)L-1}$ with $1\leq s\leq n$. Recall that $a^{(k,s)}_{vu}$ is the $v$-th row and the $u$-th column element of
$\mathcal{A}(k,s)=\mathcal{A}_{k}\cdots\mathcal{A}_{s}$.	
Thus by Remark \ref{remark1} (a) and (b), we have 
\begin{align}
	&a^{(K_{t-1}L-1,(K_{t-1}-n)L)}_{vu}\nonumber\\
	\geq& a^{(K_{t-1}L-1,(K_{t-1}-1)L)}_{vv_{1}}a^{((K_{t-1}-1)L-1,(K_{t-1}-2)L)}_{v_{1}v_{2}}\nonumber\\
	&\cdots a^{((K_{t-1}-n+1)L-1,(K_{t-1}-n)L)}_{v_{n-1}u}\nonumber\\
	\geq & \underbrace{\delta^{L} \delta^{L}\cdots  \delta^{L}}_{n}
	=  \delta^{nL}.\label{45}
\end{align}
Then for any $t\geq nL+2$ and any  $l\leq (K_{t-1}-n)L-1\leq t-nL-2$, it can be derived from  Assumption \ref{a1} (i) and (\ref{45}) that 
\begin{align}
	&a^{(t-1,l)}_{ij}=\sum^{n}_{v=1}	a^{(t-1,K_{t-1}L)}_{iv}a^{(K_{t-1}L-1,l)}_{vj}\nonumber\\
	\geq&\sum^{n}_{v=1}	a^{(t-1,K_{t-1}L)}_{iv}\min_{v\in\{1,...,n\}}a^{(K_{t-1}L-1,l)}_{vj}
	\nonumber\\
	=&\min_{v\in\{1,...,n\}}a^{(K_{t-1}L-1,l)}_{vj}\nonumber\\
	=&\min_{v\in\{1,...,n\}}\left\{\sum^n_{u=1}a^{(K_{t-1}L-1,(K_{t-1}-n)L)}_{vu}a^{((K_{t-1}-n)L-1,l)}_{uj}\right\}	\nonumber\\
	\geq&\min_{v\in\{1,...,n\}}\left\{\sum^n_{u=1}\delta^{nL}a^{((K_{t-1}-n)L-1,l)}_{uj}\right\}
	=\delta^{nL}.
\end{align}
By  (\ref{p}), we have
\begin{align}
	P^{-1}_{t,i}(t)=&\sum^n_{j=1}\sum^{t-1}_{l=0}a^{(t-1,l)}_{ij}\varphi_{l,j}(t)\varphi^{\top}_{l,j}(t)
	\!+\frac{1}{\beta}I_{p_t}\label{47}\\
	\geq& \sum^n_{j=1}\sum^{t-nL-2}_{l=0}a^{(t-1,l)}_{ij}\varphi_{l,j}(t)\varphi^{\top}_{l,j}(t)
	\!+\frac{1}{\beta}I_{p_t}\nonumber\\
	\geq& \delta^{nL}\sum^n_{j=1}\sum^{t-nL-2}_{l=0}\varphi_{l,j}(t)\varphi^{\top}_{l,j}(t)
	\!+\frac{1}{\beta}I_{p_t}.
\end{align}
Since $\delta<1$, then the following inequality
\begin{align}
	\lambda_{\min}(P^{-1}_{t}(t))&\geq \delta^{nL}\lambda_{\min}\left(\sum^n_{j=1}\sum^{t-nL-2}_{l=0}\varphi_{l,j}(t)\varphi^{\top}_{l,j}(t)
	\!+\!\frac{1}{\beta}I_{p_t}\right)\nonumber\\
	&=\delta^{nL}\lambda_{\min}(t)
	\label{49}
\end{align}
holds.
Moreover, by (\ref{47})  we have for any $i\in\{1,...,n\}$
\begin{align*}
	\lambda_{\max}(P^{-1}_{t,i}(t))&\!\leq\! \lambda_{\max}\left(\sum^n_{j=1}\sum^{t-1}_{l=0}a^{(t-1,l)}_{ij}\varphi_{l,j}(t)\varphi^{\top}_{l,j}(t)
	\!+\!\frac{1}{\beta}I_{p_t}\right)\nonumber\\
	&\leq \frac{1}{\beta}+\sum^n_{j=1}\sum^{t-1}_{l=0}\|\varphi_{l,j}(t)\|^2.
\end{align*}
Thus we have  
\begin{align}
	\lambda_{\max}(P^{-1}_{t}(t))&=\max_{i\in\{1,...,n\}}\lambda_{\max}(P^{-1}_{t,i}(t))\nonumber\\
	&\leq \frac{1}{\beta}+\sum^n_{j=1}\sum^{t-1}_{l=0}\|\varphi_{l,j}(t)\|^2=O(r_t).\label{51}
\end{align}
Hence by (\ref{49}), (\ref{51}) and Lemma \ref{th1}, it follows that
\begin{align}
	&\!\!\operatorname{tr}(\widetilde{\Theta}^{\top}_{t}(t)\widetilde{\Theta}_{t}(t))\leq\frac{\operatorname{tr}\left(\widetilde{\Theta}^{\top}_{t}(t){P}_{t}^{-1}(t)\widetilde{\Theta}_{t}(t)\right)}{\lambda_{\min}(P^{-1}_{t}(t))}\nonumber\\
	\!\!=&O\Bigg(\frac{1}{\lambda_{\min}(t)}[p_t\log^+\{\lambda_{\max}({P}_{t}^{-1}(t))\}\nonumber\\
	&~~+s_t+o([d(t)\log\log t]^2)]\Bigg)\nonumber\\
	\!\!=&O\left(\frac{1}{\lambda_{\min}(t)}\left[p_t\log r_t\!+\!s_t\! +\!o\left([d_t\log\log(t)]^2\right)\right]\right)\label{52} \text{a.s.}.
\end{align}
Moreover, we have
$
\sum^{\infty}_{q=p_t+1}\operatorname{tr}\left(\left(\Theta^{[q]}\right)^{\top}\Theta^{[q]}\right)
= O(\gamma_t).
$

Thus combining this with (\ref{52}),
by the definition of ${\hat\Theta}_{i}(t)$,  we have
\begin{align*}
	\!\!&	\|{\hat\Theta}_{i}(t)-\Theta\|^2\leq \operatorname{tr}\left(({\hat\Theta}_{i}(t)-\Theta)^{\top}({\hat\Theta}_{i}(t)-\Theta)\right)\nonumber\\
	\!\!\leq& \operatorname{tr}\left((\Theta_{t,i}(t)-\Theta(t))^{\top}(\Theta_{t,i}(t)-\Theta(t))\right)\nonumber\\
	&+\sum^{\infty}_{q=p_t+1}\operatorname{tr}\left(\left(\Theta^{[q]}\right)^{\top}\Theta^{[q]}\right)\nonumber\\
	\!\!=&O\!\left(\frac{1}{\lambda_{\min}(t)}\left[p_t\log r_t\!+\!s_t\! +\!o\!\left([d_t\log\log(t)]^2\right)\right]+\gamma_t\right) \text{a.s.}
\end{align*}
This completes the proof of the theorem.
\end{proof}

From Remark \ref{re4.2}, we can see that if in Theorem \ref{th2} the regression lag $p_t$ is of $O(t^{\alpha})$, $0<\alpha\leq 1$, then similar results can be derived.
\begin{theorem}\label{th2'}
Under the conditions of Theorem \ref{th2}, if the regression lag $p_t$ in Algorithm \ref{algorithm} is taken as $p_t=O(t^{\alpha})$, $0<\alpha\leq 1$, then (\ref{th}) hold with ``$\log\log t$" in them being replaced by $\log t$.
\end{theorem}

Under some additional conditions, the above results in Theorems \ref{th2} and \ref{th2'} can be  refined and articulated in a more concise form. Thus we can obtain the following results.
\begin{theorem}\label{re5.1}
Under the conditions of Theorem \ref{th2}, if  the following conditions hold:

(i) the noise process $\{\omega_{k,i}\}$ satisfy $\omega_{k,i}=O(\log^{1/2}k)$ for all $i$, cf., Gaussian white noise (i.i.d.); 

(ii) the unknown parameters satisfy $\Theta^{[q]}=O(\lambda^q)$ with some $0<\lambda<1$  for $q\geq 1$; 

(iii) the regression vectors satisfy $\sum^n_{i=1}\sum^{t-1}_{k=0}\big\|\varphi^{[q]}_{k,i}\big\|^2=O(t^b)$ for all $q\geq 1$ and some $b\geq 0$.

Then by taking the regression lag $p_t=[\log^{\alpha}t]$ with some $\alpha>1$,
or taking $p_t=[t^{\alpha_1}]$ with $0<\alpha_1 < 1$, we have
\begin{align}
	\|{\hat\Theta}_{i}(t)-\Theta\|^2
	=O\left(\frac{p_t\log t}{\lambda_{\min}(t)}\right) \text{a.s.,}\label{r3}
\end{align}
where $\lambda_{\min}(t)$ is defined in Theorem \ref{th2}.
\end{theorem}
\begin{proof}
By the definition of $\epsilon^{\top}_{k,i}(l)$ and Schwartz inequality,  we have the following bound for $s_t$ that
\begin{align}
	s_t=&\sum^n_{i=1}\sum^{t-1}_{k=0}\|\epsilon_{k,i}(t)\|^2\nonumber\\
	=&\sum^n_{i=1}\sum^{t-1}_{k=0}\left\|\sum^{\infty}_{q=p_t+1}\varphi^{[q]}_{k,i}\Theta^{[q]}\right\|^2\nonumber\\
	\leq&\sum^n_{i=1}\sum^{t-1}_{k=0}\left(\sum^{\infty}_{q=p_t+1}\left\|\varphi^{[q]}_{k,i}\right\|\left\|\Theta^{[q]}\right\|\right)^2\nonumber\\
	\leq\!&\sum^{\infty}_{q=p_t+1}\left\|\Theta^{[q]}\right\|\cdot
	\sum^{\infty}_{q=p_t+1}\left(\left\|\Theta^{[q]}\right\|\sum^n_{i=1}\sum^{t-1}_{k=0}\left\|\varphi^{[q]}_{k,i}\right\|^2\right)\nonumber\\
	=&O(\gamma_t t^b), \label{r4}
\end{align}
where $\gamma_t$ is  defined in Theorem \ref{th2}.
Moreover, by condition (iii) and the definition of $\lambda_{\min}(t)$, we have $\lambda_{\min}(t)=O(t^b)$.
Combining this with (\ref{r4}), we can get 
\begin{align}
	\gamma_t=O\left(\frac{\gamma_t t^b}{\lambda_{\min}(t)}\right), ~~~\frac{s_t}{\lambda_{\min}(t)}=O\left(\frac{\gamma_t t^b}{\lambda_{\min}(t)}\right).
\end{align}
By condition (ii) and  the selection of $p_t$, it can be obtained that
$\gamma_t t^b=O(\lambda^{2(p_t+1)} t^b)=O(p_t\log t)$. By condition (iii) and the definition of $r_t$ in Theorem \ref{th2}, we have $r_t=O(p_tt^b)$.
Thus by condition (i), we can see from Theorems \ref{th2} and \ref{th2'} that 
\begin{align}
	\|{\hat\Theta}_{i}(t)-\Theta\|^2
	=O\left(\frac{p_t\log t}{\lambda_{\min}(t)}\right) \text{a.s..}
\end{align}
This completes the proof of the theorem.
\end{proof}

\begin{remark}\label{re5.2}
From (\ref{r3}), we can  get the strong convergence of Algorithm \ref{algorithm}
if the following cooperative excitation condition 
\begin{align}
	\frac{p_t\log t}{\lambda_{\min}(t)}\xrightarrow{t\to\infty}0,~ \text{a.s.,}\label{e17}
\end{align}  is satisfied. We give some explanations for this condition.

(i) 
For  $p_t=[\log^{\alpha}t]$ with $\alpha>1$, the condition (\ref{e17}) is  a natural generalization of the weakest possible cooperative non-persistent excitation (PE) condition for finite dimensional regression models in \cite{Xie2021}. Notably, it is much weaker than the standard cooperative PE condition, i.e., $\liminf_{t\rightarrow\infty}\frac{\lambda_{\min}(t)}{t} > 0$. Thus, the result in \cite{Xie2021} emerges as a special case of (\ref{r3}) when $p_t$ is taken as a  fixed upper bound for the finite order of the system.
For faster-growing $p_t= [t^{\alpha_1}]$ with $0 < \alpha_1 < 1$, the strong consistency of the estimates can still be achieved under the cooperative PE condition.

(ii) 
The condition (\ref{e17}) also extends the non-cooperative excitation condition from \cite{Guo1991} and \cite{DAI2025} (i.e., $\mathcal{A} = I_{n}$) by relaxing the individual agent requirement:
\begin{gather}\label{noncoo}
	\frac{p_t\log t}{\lambda_{\min}\left(\sum^{t-1}_{k=0}\varphi_{k,i}(t)\varphi^{\top}_{k,i}(t)+\gamma I\right)}\xrightarrow{t\to\infty}0.
\end{gather}
It implies that even if any individual agent can not estimate the unknown parameter matrix accurately by the traditional non-cooperative  algorithm, the whole multi-agent system is likely to fulfill the estimation task by using the distributed  algorithm. 
\end{remark}

\subsection{Regret analysis}
Next, we analyze the accuracy of online prediction.
Note that (\ref{1}) hold  for any $t>0$ and $0\leq k\leq t$. Then  taking conditional expectations on both sides of (\ref{1}), we can obtain the best prediction for ${y}_{k+1,i}^{\top}$ of agent $i$ in the mean square sense:
\begin{gather}
\!\!\!\!\mathbb{E}({y}_{k+1,i}^{\top}|\mathscr{F}_k)\!=\!\mathbb{E}(\varphi^{\top}_{k,i}(t){\Theta(t)}\!+\!\epsilon^{\top}_{k,i}(t)|\mathscr{F}_k), ~k=0,...,t.\!\!\label{53}
\end{gather}
Replacing the unknown parameter $\Theta(t)$ in (\ref{53}) by the estimate $\Theta_{k,i}(t)$ of agent $i$ and omitting the residual term $\epsilon^{\top}_{k,i}(t)$, we can obtain the following online predictor of agent $i$ for ${y}_{k+1,i}^{\top}$:
\begin{gather}
\hat{y}_{k+1,i}^{\top}(t)=\mathbb{E}(\varphi^{\top}_{k,i}(t){\Theta_{k,i}(t)}|\mathscr{F}_k), ~k=0,...,t.\label{55}
\end{gather}
Normally, the difference between the best prediction  and the online predictor is referred to as ``regret" \cite{Xie2021}. Thus we denote the following  regret of agent $i$ at time $k$:
\begin{align*}
R_{k,i}(t)=\|\mathbb{E}({y}_{k+1,i}^{\top}|\mathscr{F}_k)-\hat{y}_{k+1,i}^{\top}(t)\|^2, ~k=0,...,t.
\end{align*}
While it may be unrealistic to expect the regret defined above to remain small for every individual $k$ due to the influence of noise, we can show that the regret is small in an average sense. Specifically, the accumulated regret from $k=0$ to $t$ is of order $o(t)$, as established in the following theorem.

\begin{theorem}\label{th3}
Consider the model (\ref{model})-(\ref{con}) and Algorithm \ref{algorithm}.  Under Assumptions \ref{a1} (i) and \ref{a2}. If $\Phi^{\top}_k(t)P_k(t)\Phi_k(t)=O(p_t)$ a.s., then the accumulated regret over networks has the following upper bound as $t\rightarrow\infty$:
\begin{align*}
	&\sum^n_{i=1}\!\sum^{t-1}_{k=0}\!R_{k,i}(t)
	\!=\\
	&\left\{
	\begin{matrix}
		\!O\!\left(p_t[s_t\!+\!p_t\log r_t]\right)\!+\!o(p_t[d(t)\log\log t]^2) &\!\!\!\text{if} ~p_t=O(\log^{\alpha}t);\\
		\!\!\!\!\!\!\!\!O\!\left(p_t[s_t\!+\!p_t\log r_t]\right)\!+\!o(p_t[d(t)\log t]^2)  & \!\!\!\!\!\!\!\!\!\!\text{if} ~p_t=O(t^{\alpha}),
	\end{matrix}
	\right.
\end{align*}
where $r_t$ is defined in Theorem \ref{th2},
and $s_t$, $d(t)$ are defined in Lemma \ref{th1}.
\end{theorem}
\begin{proof}
By Assumption \ref{a2}, it is clear that $\varphi_{k,i}(t), \epsilon^{\top}_{k,i}(t)$ and $\Theta_{k,i}(t)$ is $\mathscr{F}_k$-measurable. Thus from (\ref{53}) and (\ref{55}), we can obtain that
\begin{align}
	&\sum^n_{i=1}\sum^{t-1}_{k=0}R_{k,i}(t)=\sum^n_{i=1}\sum^{t-1}_{k=0}\|\mathbb{E}({y}_{k+1,i}^{\top}|\mathscr{F}_k)-\hat{y}_{k+1,i}^{\top}(t)\|^2\nonumber\\
	=&\sum^n_{i=1}\sum^{t-1}_{k=0}\|\varphi^{\top}_{k,i}(t){\widetilde\Theta_{k,i}(t)}+\epsilon^{\top}_{k,i}(t)\|^2\nonumber\\
	\leq& 2\sum^n_{i=1}\sum^{t-1}_{k=0}\|\varphi^{\top}_{k,i}(t){\widetilde\Theta_{k,i}(t)}\|^2+2\sum^n_{i=1}\sum^{t-1}_{k=0}\|\epsilon_{k,i}(t)\|^2.\label{56}
\end{align}
By the definition of $b_{k,i}(t)$ in Algorithm \ref{algorithm} and the condition $\Phi^{\top}_k(t)P_k(t)\Phi_k(t)=O(p_t)$, we know that 
\begin{align}
	b^{-1}_{k,i}(t)=1+\varphi^{\top}_{k,i}(t)P_{k,i}(t)\varphi_{k,i}(t)=O(p_t).\label{57}
\end{align}
Combining (\ref{57}) with Lemma \ref{th1} and (\ref{51}), we have
\begin{align}
	&\sum^n_{i=1}\sum^{t-1}_{k=0}\|\varphi^{\top}_{k,i}(t){\widetilde\Theta_{k,i}(t)}\|^2\nonumber\\
	=&{\color{red}\sum^n_{i=1}\sum^{t-1}_{k=0}}b^{-1}_{k,i}(t)b_{k,i}(t)\|\varphi^{\top}_{k,i}(t)\widetilde{\Theta}_{k,i}(t)\|^2\nonumber\\
	=&o(p_t[d(t)\log\log t]^2)+O\left(p_t[s_t+p_t\log r_t]\right), ~{\rm a.s..}\label{58}
\end{align}
Also by  Remark \ref{re4.2}, we have
\begin{align}
	&\sum^n_{i=1}\sum^{t-1}_{k=0}\|\varphi^{\top}_{k,i}(t){\widetilde\Theta_{k,i}(t)}\|^2\nonumber\\
	=&o(p_t[d(t)\log t]^2)+O\left(p_t[s_t+p_t\log r_t]\right), ~{\rm a.s..}\label{581}
\end{align}
Thus by the definition of $s_t$ in Lemma \ref{th1}, (\ref{56}), (\ref{58}) and \eqref{581}, we can obtain the desired result. This completes the proof of the theorem.
\end{proof}

Note that   $\lambda_{\max}\big(P_{k+1,i}(t)\big) \!=\!\frac{1}{\lambda_{\min}\big(P^{-1}_{k+1,i}(t)\big)}\!\leq\!\beta$ 
holds for all $k$ and $i$ by  (\ref{p}). Then the condition $\Phi^{\top}_k(t)P_k(t)\Phi_k(t) = O(p_t)$ can be satisfied when $\left\{\varphi_{k,i}^{[q]}\right\}_{k\geq 1}$ is bounded for every $q \geq 1$ and $i\in\{1,...,n\}$. 

From Theorem \ref{th3}, one can immediately deduce the following corollary.
\begin{corollary}\label{corollary1}
Let the conditions of Theorem \ref{th3} hold. If the conditions (i), (ii) and (iii) in Theorem \ref{re5.1} are also satisfied, then we have as  $t\rightarrow\infty$:
$$\sum^n_{i=1}\sum^{t-1}_{k=0}\!R_{k,i}(t)=O(p^2_t\log t).$$
\end{corollary}

We remark that the accumulated regret can achieve an order of $o(t)$ if
$p_t = [\log^\alpha t ]$ with $\alpha > 1$, or
$p_t = [ t^{\alpha_1} ]$ with $0 < \alpha_1 < \frac{1}{2}$, meaning that the averaged regret goes to zero.

To  conveniently extend the above results to distributed adaptive control problems in future work, 
we next derive the  upper bound for the  accumulated ``synchronized regret".
\begin{theorem}\label{th4}
Consider the model (\ref{model})-(\ref{con}) and Algorithm \ref{algorithm} with $p_t=[\log^\alpha t ]$ with $\alpha > 1$. 
Under Assumptions \ref{a1} (i) and \ref{a2}. If the  conditions (i), (ii) in Theorem \ref{re5.1} are  satisfied and $\left\{\varphi_{k,i}^{[q]}\right\}_{k\geq 1}$ is bounded for every $q \geq 1$ and $i\in\{1,...,n\}$,  then the following results holds as $t\rightarrow\infty$:
$$\sum^n_{i=1}\!\sum^{t-1}_{k=0}\!R_{k,i}(k)=O(p^3_t\log t).$$
\end{theorem}
\begin{proof}
Similar to (\ref{56}), we have
\begin{align}
	&\sum^n_{i=1}\sum^{t-1}_{k=0}R_{k,i}(k)
	=\sum^n_{i=1}\sum^{t-1}_{k=0}\|\varphi^{\top}_{k,i}(k){\widetilde\Theta_{k,i}(k)}+\epsilon^{\top}_{k,i}(k)\|^2\nonumber\\
	\leq& 2\sum^n_{i=1}\sum^{t-1}_{k=0}(\varphi^{\top}_{k,i}(k){\widetilde\Theta_{k,i}(k)})^2+2\sum^n_{i=1}\sum^{t-1}_{k=0}\|\epsilon_{k,i}(k)\|^2.\label{n65}
\end{align}
By (\ref{58}) and the properties of $w_{k,i}$, $\Theta^{[q]}$ and $p_t$ in Theorem \ref{re5.1}, we have for any $t\geq 1$
\begin{gather}
	\!\!\!\!\sum^n_{i=1}\sum^{t-1}_{k=0}\!\|\varphi^{\top}_{k,i}(t){\widetilde\Theta_{k,i}(t)}\|^2\!=\!O\left(p_t^2\log t\right)\!=\!O\left(\log^{2\alpha +1} t\right).\label{68}
\end{gather}

By the definition of $p_k$, we know that  $p_k=j$ is equivalent to 
$j\leq\log^\alpha k<j+1.
$
This implies that   $k\in[e^{j^{1/{\alpha}}}, e^{(j+1)^{1/{\alpha}}})$ with $e$ being the natural base  of the logarithm.
We denote the integer function $q_{j}=[e^{j^{1/{\alpha}}}]$. There exists a suitably large integer $j_0>0$ such that for any $j\geq j_0$, $q_{j+1}\in[e^{j^{1/{\alpha}}}, e^{(j+1)^{1/{\alpha}}})$, which implies $p_{q_{j+1}}=j$.

Then by the definition of $\varphi_{k,i}(k)$ and $\widetilde\Theta_{k,i}(k)$, we have for $k\in[q_{j}+1,q_{j+1}]\subset[e^{j^{1/{\alpha}}}, e^{(j+1)^{1/{\alpha}}})$, $j\geq j_0$, 
\begin{align}
	\varphi^{\top}_{k,i}(k){\widetilde\Theta_{k,i}(k)}	=&\varphi^{\top}_{k,i}(1:p_k){\widetilde\Theta_{k,i}(1:p_k)}\nonumber\\
	=&\varphi^{\top}_{k,i}(1:j){\widetilde\Theta_{k,i}(1:j)}\nonumber\\
	=&\varphi^{\top}_{k,i}(1:p_{q_{j+1}}){\widetilde\Theta_{k,i}(1:p_{q_{j+1}})}\nonumber\\
	=&	\varphi^{\top}_{k,i}(q_{j+1}){\widetilde\Theta_{k,i}(q_{j+1})}\label{69},
\end{align}
where $\varphi_{k,i}(1:m)\triangleq \text{col}\{\varphi_{k,i}^{[1]},\cdots,\varphi_{k,i}^{[m]}\}$ and $\widetilde\Theta_{k,i}(1:m)\triangleq \text{col}\{\widetilde\Theta_{k,i}^{[1]},\cdots, \widetilde\Theta_{k,i}^{[m]}\}$.

Now, for any $t>k_0$, there exists a positive integer $z$ such that $q_z<t\leq q_{z+1}$, we obtain by (\ref{68}) and (\ref{69})
\begin{align}
	&\!\!\!\!\!\sum^{t-1}_{k=k_0}\|\varphi^{\top}_{k,i}(k){\widetilde\Theta_{k,i}(k)}\|^2	\!=\!\sum^{z}_{j=j_0}\sum^{q_{j+1}}_{k=q_j+1}\!\!\!\|\varphi^{\top}_{k,i}(q_{j+1}){\widetilde\Theta_{k,i}(q_{j+1})}\|^2\nonumber\\
	=&O\left(\sum^{z}_{j=0}\log^{2\alpha +1} (q_{j+1})\!\!\right)
	\!\!\!=\!\!
	O\!\left((z+1)\log^{2\alpha +1} (q_{z+1})\right),\label{70}
\end{align}
where $k_0\geq q_{j_0}+1$.
By the definition of $q_z$, we know that $e^{z^{1/{\alpha}}}-1<q_{z}<t$, which implies that $z+1=O(p_t)$. Then by $q_{z+1}=O(q_z)=O(t)$ and (\ref{70}), it is clear that
\begin{align}
	\sum^{t-1}_{k=0}\|\varphi^{\top}_{k,i}(k){\widetilde\Theta_{k,i}(k)}\|^2=O(p^3_t\log t).\label{71}
\end{align}
By the condition that $\left\{\varphi_{k,i}^{[q]}\right\}_{k\geq 1}$ is bounded for every $q \geq 1$ and $i\in\{1,...,n\}$ and (ii) in Theorem \ref{re5.1},  we have 
\begin{align}
	&\sum^n_{i=1}\sum^{t-1}_{k=0}\|\epsilon_{k,i}(k)\|^2
	=\sum^n_{i=1}\sum^{t-1}_{k=0}\left\|\sum^{\infty}_{q=p_k+1}\varphi^{[q]}_{k,i}\Theta^{[q]}\right\|^2\nonumber\\
	\leq&\sum^n_{i=1}\sum^{t-1}_{k=0}\left(\sum^{\infty}_{q=p_k+1}\left\|\varphi^{[q]}_{k,i}\right\|\left\|\Theta^{[q]}\right\|\right)^2\nonumber\\
	=&O\left(\sum^n_{i=1}\sum^{t-1}_{k=0}\left(\sum^{\infty}_{q=p_k+1}\lambda^q\right)^2\right)
	\nonumber\\
	=&O\left(\sum^{t-1}_{k=0}\lambda^{2(p_k+1)}\right)=O(1). \label{72}
\end{align}
Thus by (\ref{n65}), (\ref{71}) and (\ref{72}), we can get the desired result. This completes the proof of the theorem.
\end{proof}

From  Theorem \ref{th3}, Corollary \ref{corollary1} and Theorem \ref{th4}, we know that our regret bounds are derived without requiring any excitation conditions on
the regression signals.
This implies that the prediction performance may not necessarily depend on estimation accuracy.

\section{Concluding Remarks}\label{section_conclusion}
This paper proposed a distributed least squares algorithm with increasing dimensions to estimate the infinite unknown parameters of a stochastic regression model. 
We analyzed the convergence properties of our algorithm under a cooperative excitation condition for the random regression vectors, without requiring assumptions of independence, stationarity, or ergodicity. This relaxation enables our theoretical results to be applicable to the feedback systems.  Moreover, we established an upper bound of the accumulated regret  without any excitation conditions. Our work generalizes prior results for finite-dimensional stochastic regression models and extends those for the single-agent case to a distributed setting.
Future research includes investigating the privacy protection and algorithm complexity problems in distributed learning for large models, exploring the combination of distributed adaptive estimation with distributed control, and investigating the estimation and prediction of nonlinear systems such as multi-layer models.

{\appendices 
	\section{Proof of Lemma \ref{lemma3}}\label{proof:Lemma4.2}
	The following lemma  establishes the relationship  on the convex
	combination of nonnegative definite matrices.
	\begin{lemma}  \cite{Xie2021} \label{LemmCoro41}
		For any adjacency matrix $\mathcal{A}_k=\{a^{(k)}_{ij}\}\in\mathbb{R}^{n\times n}$ of a weight-balanced digraph, denote $\mathscr{A}_k(t)=\mathcal{A}_k\otimes I_{p_t}$. Then for any $k\geq 1$, we have
		\begin{align*}
			\mathscr{A}^{\top}_k(t)\bar{{P}}_{k+1}^{-1}(t)\mathscr{A}_k(t)\leq {P}_{k+1}^{-1}(t),
		\end{align*}
		and
		\begin{align*}
			\mathscr{A}^{\top}_k(t){P}_{k+1}(t)\mathscr{A}_k(t)\leq \bar{{P}}_{k+1}(t).
		\end{align*}
	\end{lemma}
	
	Based on Lemma \ref{LemmCoro41}, we give the proof of Lemma \ref{lemma3} as follows.
	
	\begin{proof}
		By (\ref{error_equation}), we have
		\begin{align}
			V_{k+1}(l)=&\widetilde{\Theta}^{\top}_{k+1}(l){P}_{k+1}^{-1}(l)\widetilde{\Theta}_{k+1}(l)\nonumber\\
			=&\Big(\widetilde{\Theta}^{\top}_{k}(l){P}_{k}^{-1}(l)\mathscr{A}^{\top}_k(l){P}_{k+1}(l)-(W_{k+1}+\epsilon_{k}(l))\nonumber\\
			&~~\cdot\Phi^{\top}_{k}(l){P}_{k}(l)c_k(l){\bar P}^{-1}_{k+1}(l)\mathscr{A}^{\top}_k(l)P_{k+1}(l)
			\Big)\nonumber\\
			&\cdot\Big( {\mathscr{A}_k}(l){P}^{-1}_{k}(l)\widetilde{{\Theta}}_{k}(l)-{\mathscr{A}_k}(l){\bar P}^{-1}_{k+1}(l)\nonumber\\
			&~~~\cdot c_k(l){P}_{k}(l)\Phi_{k}(l)\left(W^{\top}_{k+1}+\epsilon^{\top}_{k}(l)\right)\Big)
			\nonumber\\
			=&\Big(\widetilde{\Theta}^{\top}_{k}(l){P}_{k}^{-1}(l)-(W_{k+1}+\epsilon_{k}(l))\nonumber\\
			&~~\cdot\Phi^{\top}_{k}(l){P}_{k}(l)c_k(l){\bar P}^{-1}_{k+1}(l)
			\Big)\mathscr{A}^{\top}_k(l){P}_{k+1}(l){\mathscr{A}_k}(l)\nonumber\\
			&\cdot\Big( {P}^{-1}_{k}(l)\widetilde{{\Theta}}_{k}(l)-{\bar P}^{-1}_{k+1}(l)\nonumber\\
			&~~~\cdot c_k(l){P}_{k}(l)\Phi_{k}(l)\left(W^{\top}_{k+1}+\epsilon^{\top}_{k}(l)\right)\Big).
			\label{v1}
		\end{align}	
		Then by (\ref{v1}) and Lemma \ref{LemmCoro41}, we have
		\begin{align}
			V_{k+1}(l)	\leq&\Big(\widetilde{\Theta}^{\top}_{k}(l){P}_{k}^{-1}(l)-(W_{k+1}+\epsilon_{k}(l))\nonumber\\
			&~~\cdot\Phi^{\top}_{k}(l){P}_{k}(l)c_k(l){\bar P}^{-1}_{k+1}(l)
			\Big){\bar P}_{k+1}(l)\nonumber\\
			&\cdot\Big( {P}^{-1}_{k}(l)\widetilde{{\Theta}}_{k}(l)-{\bar P}^{-1}_{k+1}(l)c_k(l){P}_{k}(l)\Phi_{k}(l)\nonumber\\
			&~~~\cdot \left(W^{\top}_{k+1}+\epsilon^{\top}_{k}(l)\right)\Big).\nonumber
		\end{align}
		Thus it can be obtained by taking the trace operation on both sides of the above equation, 
		\begin{align}	
			&\operatorname{tr}(V_{k+1}(l))\nonumber\\
			\leq&	\operatorname{tr}\left(\widetilde{\Theta}^{\top}_{k}(l){P}_{k}^{-1}(l){\bar P}_{k+1}(l){P}^{-1}_{k}(l)\widetilde{{\Theta}}_{k}(l)\right)\nonumber\\
			&-2\operatorname{tr}\Big(\widetilde{\Theta}^{\top}_{k}(l){P}_{k}^{-1}(l) c_k(l){P}_{k}(l)\Phi_{k}(l)[W^{\top}_{k+1}+\epsilon^{\top}_{k}(l)]
			\Big)\nonumber\\
			&+\operatorname{tr}\Big([W_{k+1}+\epsilon_{k}(l)]\Phi^{\top}_{k}(l){P}_{k}(l)c_k(l)\nonumber\\
			& ~~~\cdot{\bar P}^{-1}_{k+1}(l)c_k(l){P}_{k}(l)\Phi_{k}(l)[W^{\top}_{k+1}+\epsilon^{\top}_{k}(l)]
			\Big).\label{v2}
		\end{align}
		Now, we proceed to compute the right-hand side (RHS) of
		(\ref{v2}) term by term. Firstly, we know that from (\ref{al1})
		\begin{align}	
			&\operatorname{tr}\left(\widetilde{\Theta}^{\top}_{k}(l){P}_{k}^{-1}(l)\bar{P}_{k+1}(l){P}^{-1}_{k}(l)\widetilde{{\Theta}}_{k}(l)\right)\nonumber\\
			=&\operatorname{tr}\Big(\widetilde{\Theta}^{\top}_{k}(l){P}_{k}^{-1}(l)\Big({P}_{k}(l)-c_k(l){P}_{k}(l)\Phi_{k}(l)\nonumber\\
			&~~\cdot\Phi^{\top}_{k}(l)P_{k}(l)\Big){P}^{-1}_{k}(l)\widetilde{{\Theta}}_{k}(l)\Big)\nonumber\\
			=&\operatorname{tr}(V_k(l))-\operatorname{tr}\left(\widetilde{\Theta}^{\top}_{k}(l)\Phi_{k}(l)b_k(l)\Phi^{\top}_{k}(l)\widetilde{{\Theta}}_{k}(l)
			\right),\label{v3}
		\end{align}
		where the last inequality holds since we have
		\begin{align}\label{27}
			\left\{
			\begin{aligned}	&c_k(l)P_{k}(l)=P_{k}(l)c_k(l),~\Phi^{\top}_{k}(l)c_k(l)=b_k(l)\Phi^{\top}_{k}(l),\\
				&	c_k(l)\Phi_{k}(l)=\Phi_{k}(l)b_k(l).\\
			\end{aligned}
			\right.
		\end{align}
		Then for the second term on the RHS of (\ref{v2}), we have  by (\ref{27})
		\begin{align}
			&\operatorname{tr}\Big(\widetilde{\Theta}^{\top}_{k}(l){P}_{k}^{-1}(l) c_k(l){P}_{k}(l)\Phi_{k}(l)[W^{\top}_{k+1}+\epsilon^{\top}_{k}(l)]
			\Big)\nonumber\\
			=&\operatorname{tr}\Big(\widetilde{\Theta}^{\top}_{k}(l) \Phi_{k}(l)b_k(l)[W^{\top}_{k+1}+\epsilon^{\top}_{k}(l)]
			\Big).\label{v4}   
		\end{align}
		By (\ref{bar}), we have 
		\begin{align}
			\bar P^{-1}_{k+1}(l)= P^{-1}_{k}(l)+\Phi_{k}(l)\Phi^{\top}_{k}(l).\label{bar2}
		\end{align}
		For the last term on the RHS of (\ref{v2}), by Lemma \ref{LemmCoro41} and (\ref{bar2}), we compute it as follows:
		\begin{align}
			&\operatorname{tr}\Big([W_{k+1}+\epsilon_{k}(l)]\Phi^{\top}_{k}(l){P}_{k}(l)c_k(l)\nonumber\\
			& ~~~\cdot {\bar P}^{-1}_{k+1}(l)c_k(l){P}_{k}(l)\Phi_{k}(l)[W^{\top}_{k+1}+\epsilon^{\top}_{k}(l)]
			\Big)\nonumber\\
			=&\operatorname{tr}\Big([W_{k+1}+\epsilon_{k}(l)]\Phi^{\top}_{k}(l){P}_{k}(l)c_k(l)[P^{-1}_{k}(l)+\Phi_{k}(l)\Phi^{\top}_{k}(l)]\nonumber\\
			& ~~~c_k(l){P}_{k}(l)\Phi_{k}(l)[W^{\top}_{k+1}+\epsilon^{\top}_{k}(l)]
			\Big)
			\nonumber\\
			=&\operatorname{tr}\Big([W_{k+1}+\epsilon_{k}(l)]b^2_k(l)\Phi^{\top}_{k}(l){P}_{k}(l)\Phi_{k}(l)[W^{\top}_{k+1}+\epsilon^{\top}_{k}(l)]
			\Big)	\nonumber\\
			&+\operatorname{tr}\Big([W_{k+1}+\epsilon_{k}(l)]b^2_k(l)\Phi^{\top}_{k}(l){P}_{k}(l)\Phi_{k}(l)\Phi^{\top}_{k}(l)\nonumber\\
			&~~~\cdot {P}_{k}(l)\Phi_{k}(l)[W^{\top}_{k+1}+\epsilon^{\top}_{k}(l)]
			\Big)\nonumber\\	=&\operatorname{tr}\Big([W_{k+1}+\epsilon_{k}(l)]b_k(l)\Phi^{\top}_{k}(l){P}_{k}(l)\Phi_{k}(l)[W^{\top}_{k+1}+\epsilon^{\top}_{k}(l)]
			\Big).\label{v5}
		\end{align}
		Thus by (\ref{v2}), (\ref{v3}), (\ref{v4}) and (\ref{v5}), we can obtain that
		\begin{align}
			&\operatorname{tr}(V_{k+1}(l))\nonumber\\
			\leq& \operatorname{tr}(V_k(l))-\operatorname{tr}\left(\widetilde{\Theta}^{\top}_{k}(l)\Phi_{k}(l)b_k(l)\Phi^{\top}_{k}(l)\widetilde{{\Theta}}_{k}(l)
			\right)\nonumber\\
			&-2\operatorname{tr}\Big(\widetilde{\Theta}^{\top}_{k}(l)\Phi_{k}(l)b_k(l)	[W^{\top}_{k+1}\!+\!\epsilon^{\top}_{k}(l)]\Big)\nonumber\\
			&+\operatorname{tr}\Big([W_{k+1}+\epsilon_{k}(l)]b_k(l)\Phi^{\top}_{k}(l){P}_{k}(l)\Phi_{k}(l)[W^{\top}_{k+1}+\epsilon^{\top}_{k}(l)]
			\Big).
		\end{align}
		Summing from $k=0$ to $t-1$ yields that 
		\begin{align}
			&\operatorname{tr}(V_{t}(l))+\sum^{t-1}_{k=0}\operatorname{tr}\left(\widetilde{\Theta}^{\top}_{k}(l)\Phi_{k}(l)b_k(l)\Phi^{\top}_{k}(l)\widetilde{{\Theta}}_{k}(l)
			\right)\nonumber\\
			\leq & \operatorname{tr}(V_0(l))\!-\!\!2\!\sum^{t-1}_{k=0}\operatorname{tr}\Big(\widetilde{\Theta}^{\top}_{k}(l)
			\Phi_{k}(l)b_k(l)[W^{\top}_{k+1}\!+\!\epsilon^{\top}_{k}(l)]\Big)\nonumber\\
			&+\sum^{t-1}_{k=0}\operatorname{tr}\Big([W_{k+1}+\epsilon_{k}(l)]b_k(l)\Phi^{\top}_{k}(l){P}_{k}(l)\nonumber\\
			&~~~~~~~~~~~\cdot\Phi_{k}(l)[W^{\top}_{k+1}+\epsilon^{\top}_{k}(l)]
			\Big).\label{v6}
		\end{align}

		By  the properties of the trace, we can  obtain that
		\begin{align}
			&\sum^{t-1}_{k=0}\operatorname{tr}\Big(\widetilde{\Theta}^{\top}_{k}(l)
			\Phi_{k}(l)b_k(l)W^{\top}_{k+1}\Big)\nonumber\\
			=&\sum^n_{i=1}\sum^{t-1}_{k=0}\operatorname{tr}\Big(\widetilde{\Theta}^{\top}_{k,i}(l)
			\varphi_{k,i}(l)b_{k,i}(l)w^{\top}_{k+1,i}\Big)\nonumber\\
			=&\sum^n_{i=1}\sum^{t-1}_{k=0}b_{k,i}(l)\varphi^{\top}_{k,i}(l)\widetilde{\Theta}_{k,i}(l)w_{k+1,i}.\label{new4}
		\end{align}
		
		By the fact $b_k(l)\leq I_{np_l}$ , it is clear that 
		\begin{align}
			&-\sum^{t-1}_{k=0}\operatorname{tr}\Big(\widetilde{\Theta}^{\top}_{k}(l)
			\Phi_{k}(l)b_k(l)\epsilon^{\top}_{k}\Big)\nonumber\\
			\leq& \frac{1}{4}\sum^{t-1}_{k=0}\operatorname{tr}\left(\widetilde{\Theta}^{\top}_{k}(l)\Phi_{k}(l)b_k(l)\Phi^{\top}_{k}(l)\widetilde{{\Theta}}_{k}(l)
			\right)\!+\!\sum^{t-1}_{k=0}\operatorname{tr}\left(\epsilon^{\top}_{k}\epsilon_{k}\right)\nonumber\\
			=& \frac{1}{4}\sum^{t-1}_{k=0}\operatorname{tr}\left(\widetilde{\Theta}^{\top}_{k}(l)\Phi_{k}(l)b_k(l)\Phi^{\top}_{k}(l)\widetilde{{\Theta}}_{k}(l)
			\right)\!\!+\!\!\sum^n_{i=1}\!\sum^{t-1}_{k=0}\|\epsilon_{k,i}(l)\|^2.\label{new5}
		\end{align}
		By the proof  of Lemma 4.4 in \cite{Xie2021} (see (33) in \cite{Xie2021}), we have 
		\begin{align}
			&\lambda_{\max}(b_k(l)\Phi^{\top}_{k}(l){P}_{k}(l)\Phi_{k}(l))\nonumber\\
			\leq&\frac{{\rm det}\left({P}_{k+1}^{-1}(l)\right)-{\rm det}\left({P}_{k}^{-1}(l)\right)}{{\rm det}\left({P}_{k+1}^{-1}(l)\right)}\leq 1.\label{new7}
		\end{align}
		Then it follows from (\ref{new7}) that
		\begin{align}
			&\sum^{t-1}_{k=0}\operatorname{tr}\Big([W_{k+1}\!+\!\epsilon_{k}(l)] 
			b_k(l)\Phi^{\top}_{k}(l){P}_{k}(l)\Phi_{k}(l)[W^{\top}_{k+1}\!+\!\epsilon^{\top}_{k}(l)]
			\Big)\nonumber\\
			\leq& 2\sum^{t-1}_{k=0}\operatorname{tr}\Big(W_{k+1}
			b_k(l)\Phi^{\top}_{k}(l){P}_{k}(l)\Phi_{k}(l)W^{\top}_{k+1}
			\Big)\nonumber\\
			&+ 2\sum^{t-1}_{k=0}\operatorname{tr}\Big(\epsilon_{k}(l)
			b_k(l)\Phi^{\top}_{k}(l){P}_{k}(l)\Phi_{k}(l)\epsilon^{\top}_{k}(l)
			\Big)\nonumber\\
			\leq& 2\sum^{t-1}_{k=0}\operatorname{tr}\Big(W_{k+1}
			b_k(l)\Phi^{\top}_{k}(l){P}_{k}(l)\Phi_{k}(l)W^{\top}_{k+1}
			\Big)\nonumber\\
			&+ 2\sum^{t-1}_{k=0}\lambda_{\max}(b_k(l)\Phi^{\top}_{k}(l){P}_{k}(l)\Phi_{k}(l))\operatorname{tr}\Big(\epsilon_{k}(l)
			\epsilon^{\top}_{k}(l)
			\Big)\nonumber\\
			\leq& 2\sum^{t-1}_{k=0}\operatorname{tr}\Big(W_{k+1}
			b_k(l)\Phi^{\top}_{k}(l){P}_{k}(l)\Phi_{k}(l)W^{\top}_{k+1}
			\Big)\nonumber\\
			&+ 2\sum^n_{i=1}\!\sum^{t-1}_{k=0}\|\epsilon_{k,i}(l)\|^2.\label{36}
		\end{align}
		Substituting (\ref{new4}), (\ref{new5}) and (\ref{36}) into (\ref{v6}), we complete the proof of this lemma.
	\end{proof}
	
	\section{Proof of Lemma \ref{th1}}\label{proof:th1}
	\begin{proof}
		From Lemma \ref{lemma3}, we  estimate the last two terms with $p_t$ on the RHS of (\ref{new6}) to prove this lemma.
		
		1) Estimate  $\sum^n_{i=1}\sum^{t-1}_{k=0}b_{k,i}(t)\varphi^{\top}_{k,i}(t)\widetilde{\Theta}_{k,i}(t)w_{k+1,i}$.
		
		By Lemma \ref{lemma2} and the fact $b_{k,i}(t)\leq 1$, we have
		\begin{align}
			&\left|\max_{1\leq p_l\leq p_t}\sum^n_{i=1}\sum^{t-1}_{k=0}b_{k,i}(l)\varphi^{\top}_{k,i}(l)\widetilde{\Theta}_{k,i}(l)w_{k+1,i}\right|\nonumber\\
			=&\sum^n_{i=1}O(h_{t,i}\log h_{t,i})+\sum^n_{i=1}o(h_{t,i}d_i(t)\log\log t)\nonumber\\
			=&\sum^n_{i=1}o(h^2_{t,i})+\sum^n_{i=1}o([d_i(t)\log\log t]^2)+O(1)\nonumber\\
			=&o\left(\sum^n_{i=1}h^2_{t,i}\right)+o([d(t)\log\log t]^2)+O(1),\label{35} ~~{\rm a.s.}
		\end{align}
		where 
		\begin{align}
			h_{t,i}=\left\{\max_{1\leq p_l\leq p_t}\sum^{t-1}_{k=0}b_{k,i}(l)\|\varphi^{\top}_{k,i}(l)\widetilde{\Theta}_{k,i}(l)\|^2\right\}^{\frac{1}{2}}.\label{hti}
		\end{align}

		2) Estimate $\sum^{t-1}_{k=0}\operatorname{tr}\Big(W_{k+1}
		b_k(t)\Phi^{\top}_{k}(t){P}_{k}(t)\Phi_{k}(t)W^{\top}_{k+1}
		\Big) $.
		
		By (\ref{new7}) and $P^{-1}_{0}(l)=\frac{1}{\beta}I_{np_l}\in\mathbb{R}^{np_l\times np_l}$, we have
		\begin{align}
			&\max_{1\leq p_l\leq p_t}\sum^{t-1}_{k=0}\lambda_{\max}(b_k(l)\Phi^{\top}_{k}(l){P}_{k}(l)\Phi_{k}(l))\nonumber\\
			\leq &\max_{1\leq p_l\leq p_t}\sum^{t-1}_{k=0}\frac{{\rm det}\left({P}_{k+1}^{-1}(l)\right)-{\rm det}\left({P}_{k}^{-1}(l)\right)}{{\rm det}\left({P}_{k+1}^{-1}(l)\right)}\nonumber\\
			=&\max_{1\leq p_l\leq p_t}\sum^{t-1}_{k=0}\int^{{\rm det}\left({P}_{k+1}^{-1}(l)\right)}_{{\rm det}\left({P}_{k}^{-1}(l)\right)}\frac{1}{{\rm det}\left({P}_{k+1}^{-1}(l)\right)}dx\nonumber\\
			\leq& \max_{1\leq p_l\leq p_t}\int^{{\rm det}\left({P}_{t}^{-1}(l)\right)}_{{\rm det}\left({P}_{0}^{-1}(l)\right)}\frac{1}{x}dx\nonumber\\
			=&\max_{1\leq p_l\leq p_t}\left(\log({\rm det}\left({P}_{t}^{-1}(l)\right))-\log({\rm det}\left({P}_{0}^{-1}(l)\right))\right)
			\nonumber\\
			=&\max_{1\leq p_l\leq p_t}\left(\log({\rm det}\left({P}_{t}^{-1}(l)\right))+np_l\log\beta\right)\nonumber\\
			=&O\left(p_t\log^+\{\lambda_{\max}({P}_{t}^{-1}(t))\}\right), ~~{\rm a.s.}\label{41}
		\end{align}
		Thus from Lemma \ref{lemma1} (ii) and Assumption \ref{a2} and (\ref{41}) it follows that
		\begin{align}
			&	\max_{1\leq p_l\leq p_t}\sum^{t-1}_{k=0}\operatorname{tr}\Big(W_{k+1}
			b_k(l)\Phi^{\top}_{k}(l){P}_{k}(l)\Phi_{k}(l)W^{\top}_{k+1}
			\Big)\nonumber\\	
			\leq &\max_{1\leq p_l\leq p_t}\sum^{t-1}_{k=0}\!\!\left(\lambda_{\max}(b_k(l)\Phi^{\top}_{k}(l){P}_{k}(l)\Phi_{k}(l))\!\sum^n_{i=1}\!\|w_{k+1,i}\|^2\right)\nonumber\\	
			\leq &\max_{1\leq p_l\leq p_t}\Bigg\{\sum^{t-1}_{k=0}\Big(\lambda_{\max}(b_k(l)\Phi^{\top}_{k}(l){P}_{k}(l)\Phi_{k}(l))\nonumber\\	
			&~~~~~~~~~~~~~~~~~\cdot\sum^n_{i=1}\!(\|w_{k+1,i}\|^2-\mathbb{E}[\|w_{k+1,i}\|^{2}|\mathscr{F}_k]) \Big)\nonumber\\	
			+&\sum^{t-1}_{k=0}\lambda_{\max}(b_k(l)\Phi^{\top}_{k}(l){P}_{k}(l)\Phi_{k}(l))\sum^n_{i=1}\mathbb{E}[\|w_{k+1,i}\|^{2}|\mathscr{F}_k]) 
			\Bigg\}\nonumber\\	
			=&O\left(\max_{1\leq p_l\leq p_t}\sum^{t-1}_{k=0}\lambda_{\max}(b_k(l)\Phi^{\top}_{k}(l){P}_{k}(l)\Phi_{k}(l))\right)\nonumber\\	
			&+o(d^2(t)\log\log t)\nonumber\\	
			=&O\left(p_t\log^+\{\lambda_{\max}({P}_{t}^{-1}(t))\}\right)
			\!+\!o(d^2(t)\log\log t)~{\rm a.s..}\label{42}
		\end{align}
		Combining (\ref{new6}), (\ref{35}) and (\ref{42}), and since 
		\begin{align}
			&\sum^{t-1}_{k=0}\operatorname{tr}\left(\widetilde{\Theta}^{\top}_{k}(l)\Phi_{k}(l)b_k(l)\Phi^{\top}_{k}(l)\widetilde{{\Theta}}_{k}(l)\right)\nonumber\\
			=&\sum^n_{i=1}\sum^{t-1}_{k=0}\operatorname{tr}\left(\widetilde{\Theta}^{\top}_{k,i}(l)\varphi_{k,i}(l)b_{k,i}(l)\varphi^{\top}_{k,i}(l)\widetilde{{\Theta}}_{k,i}(l)
			\right),\nonumber\\
			=&\sum^n_{i=1}\sum^{t-1}_{k=0}b_{k,i}(l)\|\varphi^{\top}_{k,i}(l)\widetilde{{\Theta}}_{k,i}(l)\|^2,
		\end{align}
		we have 
		\begin{align}
			&\max_{1\leq p_l\leq p_t}\operatorname{tr}(V_{t}(l))+\frac{1}{2}\sum^n_{i=1}h^2_{t,i}\nonumber\\
			= & o\left(\sum^n_{i=1}h^2_{t,i}\right)+o([d(t)\log\log t]^2)+O(1)	\nonumber\\
			&+O\left(p_t\log^+\{\lambda_{\max}({P}_{t+1}^{-1}(t))\}\right)
			\!+\!o(d^2(t)\log\log t)\nonumber\\
			&+4\sum^n_{i=1}\!\sum^{t-1}_{k=0}\|\epsilon_{k,i}(t)\|^2\nonumber\\
			=&O\left(\sum^n_{i=1}\!\sum^{t-1}_{k=0}\|\epsilon_{k,i}(t)\|^2\right)+ o\left(\sum^n_{i=1}h^2_{t,i}\right)\nonumber\\
			&+o([d(t)\log\log t]^2)+O\left(p_t\log^+\{\lambda_{\max}({P}_{t}^{-1}(t))\}\right)\nonumber ~{\rm a.s.}
		\end{align}
		with $h_{t,i}$ being defined in (\ref{hti}).
		This implies that
		\begin{align}
			&\operatorname{tr}(V_{t}(t))+\left(\frac{1}{2}+o(1)\right)\sum^n_{i=1}h^2_{t,i}\nonumber\\
			=&O\left(\sum^n_{i=1}\!\sum^{t-1}_{k=0}\|\epsilon_{k,i}(t)\|^2\right) +o([d(t)\log\log t]^2)\nonumber\\
			&+O\left(p_t\log^+\{\lambda_{\max}({P}_{t}^{-1}(t))\}\right),\nonumber ~{\rm a.s.,}
		\end{align}
		which concludes the result of Lemma \ref{th1}.
	\end{proof}
}

\bibliographystyle{IEEEtran}
\bibliography{my_bib}

\end{document}